\documentclass[10pt,twocolumn,twoside]{IEEEtran}

\title{\LARGE Coherence and Concentration in Tightly-Connected Networks}
\author{Hancheng Min, Richard Pates and Enrique Mallada\thanks{H. Min and E. Mallada are with the Department of Electrical and Computer Engineering, Johns Hopkins University, Baltimore, MD 21218, USA. E-mail: \texttt{\{hanchmin, mallada\}@jhu.edu}.

Richard Pates is with the Department of Automatic Control, Lund University, Box 118, SE-221 00, Lund, Sweden. E-mail: \texttt{richard.pates@control.lth.se}. He is a member of the ELLIIT Strategic Research Area at Lund University.}
\thanks{Preliminary version of this work, covering an alternative version of the results in Section \ref{sec:dymC}, was presented in \cite{min2019cdc}.
}
}

\usepackage{mysty}
\usepackage{bbold,amsmath,amsthm,amsfonts,amssymb}
\usepackage{url}
\allowdisplaybreaks
\usepackage{environ}
\usepackage{cite}
\usepackage{accents}
\usepackage{balance}
\usepackage{ifthen}
\newboolean{showcomments}
\setboolean{showcomments}{true}
\usepackage[usenames,dvipsnames]{xcolor}
\usepackage{todonotes}

\definecolor{bleudefrance}{rgb}{0.19, 0.55, 0.91}
\definecolor{ao(english)}{rgb}{0.0, 0.5, 0.0}

\newcommand{\addcite}[0]{\ifthenelse{\boolean{showcomments}}
{\textcolor{purple}{(add cite(s)) }}{}}%

\newcommand{\enrique}[1]{  \ifthenelse{\boolean{showcomments}}
{\todo[inline,color=bleudefrance]{Enrique: #1}}{}}
\newcommand{\emmargin}[1]{\ifthenelse{\boolean{showcomments}}{\marginpar{\color{bleudefrance}\tiny EM: #1}}{}}
\newcommand{\hancheng}[1]{  \ifthenelse{\boolean{showcomments}}
{\todo[inline,color=orange]{Hancheng: #1}}{}}

\newboolean{showedits}
\setboolean{showedits}{true}
\usepackage[xcolor=bleudefrance]{changes}
\definechangesauthor[color=bleudefrance]{EM}
\newcommand{\aem}[1]{
\ifthenelse{\boolean{showedits}}
{\added[id=EM]{#1}}
{#1}
}
\newcommand{\chem}[2]{
\ifthenelse{\boolean{showedits}}
{\replaced[id=EM]{#1}{#2}}
{#1}
}
\newcommand{\dem}[1]{
\ifthenelse{\boolean{showedits}}
{\deleted[id=EM]{#1}}
{}
}
\newcommand{\hl}[1]{
\ifthenelse{\boolean{showedits}}
{{\color{bleudefrance}#1}}
{#1}
}

\setboolean{showedits}{false}
\setboolean{showcomments}{false}
\newboolean{archive}
\setboolean{archive}{true}
\newboolean{bio}
\setboolean{bio}{true}
\newboolean{color}
\setboolean{color}{false}
\usepackage{verbatim}
\usepackage{graphicx}

\newif\ifshownotes
\shownotestrue
\definecolor{notetext}{rgb}{0.7,0,0}

\newtheorem{thm}{Theorem}

\newtheorem{lem}{Lemma}

\newtheorem{prop}[thm]{Proposition}
\newtheorem{example}{Example}
\newtheorem{clm}{Claim}
\newtheorem{dfn}{Definition}
\newtheorem{asmp}{Assumption}
\newtheorem{rem}{Remark}

\begin{document}
\maketitle
\thispagestyle{plain}
\pagestyle{plain}
\begin{abstract}
    The ability to achieve coordinated behavior---engineered or emergent---on networked systems has attracted widespread interest over several fields. This interest has led to remarkable advances in developing a theoretical understanding of the conditions under which agents within a network can reach an agreement (consensus) or develop coordinated behavior, such as synchronization. However, much less understood is the phenomenon of network coherence. Network coherence generally refers to nodes' ability in a network to have a similar dynamic response despite heterogeneity in their individual behavior. 
    In this paper, we develop a general framework to analyze and quantify the level of network coherence that a system exhibits by relating coherence with a low-rank property of the system. More precisely, for a networked system with linear dynamics and coupling, we show that, as the network connectivity grows, the system transfer matrix converges to a rank-one transfer matrix representing the coherent behavior. Interestingly, the non-zero eigenvalue of such a rank-one matrix is given by the harmonic mean of individual nodal dynamics, and we refer to it as the coherent dynamics. Our analysis unveils the frequency-dependent nature of coherence and a non-trivial interplay between dynamics and network topology. 
    We further show that many networked systems can exhibit similar coherent behavior by establishing a concentration result in a setting with randomly chosen individual nodal dynamics.
\end{abstract}

\section{Introduction}\label{sec:intro}
Coordinated behavior in network systems has been a popular subject of research in many fields, including physics~\cite{Bressloff1999}, chemistry~\cite{Kiss2002}, social sciences~\cite{DeGroot1974}, and biology~\cite{Mirollo1990}. Within engineering, coordination is essential for the proper operation of many networked systems including power networks~\cite{jpm2017cdc,Paganini2019tac}, data and sensor networks~\cite{mmhzt2015ton,m2014phd-thesis}, and autonomous transportation~\cite{Sepulchre2008706,Olfati-Saber2007,Jadbabaie2003988,Bamieh2012}.
While there exist many expressions of this behavior, two forms of coordination have particularly received thorough attention by the control community: Consensus and synchronization.

Consensus~\cite{DeGroot1974, Olfati-Saber2007, Jadbabaie2003988, Bamieh2012, Tegling19,Olfati-Saber20041520, Ghaedsharaf2019}, on one hand, refers to the ability of the network nodes to asymptotically reach a common value over some quantities of interest.
Many extensions of this problem include the study of robustness and performance of consensus networks in the presence of noise~\cite{Jadbabaie2003988,Bamieh2012,Tegling19}, time-delay~\cite{Olfati-Saber20041520, Ghaedsharaf2019}, and switching graph topology~\cite{Ghaedsharaf2019}.
Synchronization~\cite{Mirollo1990,mmhzt2015ton,m2014phd-thesis,Sepulchre2008706,Nair2008661,Kim2011200,Wieland2011}, on the other hand, refers to the ability of network nodes to follow a commonly defined trajectory. Although for nonlinear systems synchronization is a structurally stable phenomenon, in the linear case~\cite{Nair2008661, Sepulchre2008706, Kim2011200, Wieland2011}, synchronization requires the existence of a common internal model that acts as a virtual leader~\cite{Kim2011200,Wieland2011}.

A closely related notion of coordination emerges when looking at how the network agents collectively respond to disturbances. In this setting, agents with noticeably different input-output responses, can present a similar, i.e., coherent, response when interconnected.
A vast body of work, triggered by the seminal paper~\cite{Bamieh2012}, has quantitatively studied the role of the network topology in the emergence of coherence. Examples include, directed~\cite{Tegling19} and undirected~\cite{Oral17} consensus networks, transportation networks~\cite{Bamieh2012}, and power networks~\cite{Paganini2019tac,Bamieh2013,Andreasson17,psf17cdc}. 
The key technical approach amounts to quantify the level of coherence by computing the $\mathcal{H}_2$-norm of the system for appropriately defined nodal disturbance and performance signals. Broadly speaking, the analysis shows a reciprocal dependence between the performance metrics and the non-zero eigenvalues of the network graph Laplacian, validating the fact that strong network coherence (low $\mathcal{H}_2$-norm) results from the high connectivity of the network (large Laplacian eigenvalues).
Unfortunately, the analysis strongly relies on a homogeneity~\cite{Bamieh2012,Tegling19, Oral17,Bamieh2013, Andreasson17,psf17cdc} or proportionality~\cite{Paganini2019tac} assumption of the nodal transfer functions, and thus fails to characterize how individual heterogeneous node dynamics affect the overall coherent network response.

In this paper, we seek to overcome these limitations by formalizing  network coherence by the presence of a low-rank structure, of the system transfer matrix, that appears when the network feedback gain is high. 
More precisely, we show that for linear networks, as the network's effective algebraic connectivity (a frequency-dependent notion to be later defined) grows, the system transfer matrix converges to a rank-one transfer matrix with identical element-wise transfer functions. Interestingly, this transfer function is given by the harmonic mean of individual nodal dynamics, which we refer to as coherent dynamics.  Furthermore, we provide the concentration result of such coherent dynamics in large-scale stochastic networks where the node dynamics are given by random transfer functions. 


This frequency domain analysis provides a deeper characterization of the role of both, network topology and node dynamics, on the coherent behavior of the network. In particular, our results make substantial contributions towards the understanding of coordinated and coherent behavior of network systems in many ways:
\begin{itemize}
    \item We present a general framework in frequency domain to analyze the coherence of heterogeneous networks. We show that network coherence emerges as a low rank structure of the system transfer matrix as we increase the effective algebraic connectivity--a frequency-varying quantity that depends on the network coupling strength and  dynamics.
    \item Unlike previous works, our analysis applies to networks with  heterogeneous nodal dynamics, and further provides an explicit characterization in the frequency domain of the coherent response to disturbances. Thus, in this way, our results highlight the contribution of individual nodal dynamics to the network's coherent behavior.
    \item The analysis further suggests that network coherence is a frequency-dependent phenomenon. That is, the ability for nodes to respond coherently depends on the frequency composition of the input disturbance. {\ifthenelse{\boolean{color}}{\color{blue}}{} We further explicitly connect our frequency domain analysis to time domain notions of coherence.}
    \item By providing an exact characterization of network's coherent dynamics, our analysis can be further applied in settings where only distributional information of the netwok composition is known. More precisely, we show that the coherent dynamics of tightly-connected networks with possibly random nodal dynamics are well approximated by a deterministic transfer function that only depends on the statistical distribution of node dynamics. 
\end{itemize}

{\ifthenelse{\boolean{color}}{\color{blue}}{}Prior works on coherence primarily study homogeneous networks~\cite{Bamieh2012,Andreasson17,psf17cdc} and proportional networks~\cite{Paganini2019tac}, where the network can be diagonalized into $n$ individual dynamics,  each corresponding to one eigenvalue of the network Laplacian. However, such techniques to study coherence cannot be generalized to non-proportional heterogeneous networks. In fact, in such setting, it is a priori unclear what a good representation of the coherent response is. Our new frequency domain framework allows one to analyze coherence in general heterogeneous networks.

Notably, the problem of characterizing coherent dynamics is unique to heterogeneous networks, since the coherent dynamics for homogeneous networks are exactly equal to the common nodal dynamics. In real applications, however,  such as power networks, such characterization is relevant to model reduction~\cite{Germond1978,Apostolopoulou2016,Guggilam2018} and control design~\cite{jbvm2021lcss}. Our analysis provides, in the asymptotic sense, the exact characterization of coherent dynamics that can be used in control design for heterogeneous networks.
}

The paper is organized as follows. In Section \ref{sec:ptw_conv} and Section \ref{sec:unifm_conv} we present respectively the point-wise and uniform convergence results of network transfer matrix as the network connectivity increases. In Section \ref{sec:dymC}, the dynamics concentration in large-scale networks is discussed. In Section \ref{sec:exmples}, we apply our analysis to consensus networks and synchronous generator networks. Conclusions are presented in Section \ref{sec:conclusion}. 

\emph{Notation:}~For a vector $x$, $\|x\|=\sqrt{x^Tx}$ denotes the $2$-norm of $x$, and for a matrix $A$, $\sigma_i(A)$ denotes the $i$th smallest singular value of $A$, $\|A\|$ denotes the spectral norm of $A$. Particularly, if $A$ is real symmetric, we let $\lambda_i(A)$ denote the $i$th smallest eigenvalue of $A$. For two sets $S_1,S_2$, we let $S_1\setminus{S_2}$ denote the set difference.

We let $I_n$ denote the identity matrix of order $n$, $\one$ denote column vector $[1,\cdots,1]^T $, $[n]$ denote the set $\{1,2,\cdots,n\}$ and $\mathbb{N}_+$ denote the set of positive integers. Also, we write complex numbers as $a+jb$, where $j=\sqrt{-1}$.

\section{Problem Setup}\label{sec:prem}
Consider a network consisting of $n$ nodes ($n\geq 2$), indexed by $i\in[n]$ with the block diagram structure in Fig.\ref{blk_p_n}. $L$ is the Laplacian matrix of the weighted graph that describes the network interconnection. We further use $f(s)$ to denote the transfer function representing the dynamics of network coupling, and $G(s)=\mathrm{diag}\{g_i(s)\}$ to denote the nodal dynamics, with $g_i(s),\ i\in[n]$, being an SISO transfer function representing the dynamics of node $i$. 
\begin{figure}[ht]
    \centering
	\includegraphics[height=2.5cm]{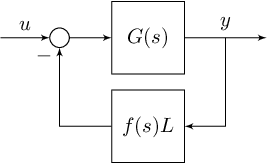}
	\caption{Block diagram of general networked dynamical systems}\label{blk_p_n}
\end{figure}

Under this setting, we can compactly express the transfer matrix from the input signal vector $u$ to the output signal vector $y$ by
\begin{align}
    T(s)&=\; (I_n+G(s)f(s)L)^{-1}G(s)\nonumber\\
    &=\; (I_n+\mathrm{diag}\{g_i(s)\}f(s)L)^{-1}\mathrm{diag}\{g_i(s)\}\,.\label{eq_T_explict}
\end{align}

Many existing networks can be represented by this structure. For example, for the first-order consensus network~\cite{Olfati-Saber20041520,Olfati-Saber2007}, $f(s)=1$, and the node dynamics are given by $g_i(s)=\frac{1}{s}$. For power networks~\cite{Andreasson17,Paganini2019tac}, $f(s)=\frac{1}{s}$, $g_i(s)$ are the dynamics of the generators, and $L$ is the Laplacian  matrix representing the sensitivity of power injection w.r.t. bus phase angles. Finally, in transportation networks~\cite{Jadbabaie2003988,Olfati-Saber2007}, $g_i(s)$ represent the vehicle dynamics whereas $f(s)L$ describes local inter-vehicle information transfer. 

Throughout this paper, we make the following assumptions, all of which are mild and commonly satisfied by several models that analyze the above-mentioned applications.
\begin{asmp}
    The closed-loop system \eqref{eq_T_explict} satisfies:
    \begin{enumerate}
        \item All $g_i(s),\ i=1,\cdots,n$ and $f(s)$ are rational proper transfer functions;
        \item Laplacian matrix $L$ is real symmetric;
        \item Any pole of $f(s)$ is not a zero of any of $g_i(s),\ i=1,\cdots,n$.
    \end{enumerate}
\end{asmp}

A straight forward application of the symmetry assumption in $L$ comes from its eigendecomposition 
\begin{equation}\label{laplcian_decomp}
    L=V\Lambda V^T\,,
\end{equation} where $V=\lhp\frac{\one}{\sqrt{n}},V_\perp\rhp$, $VV^T=V^TV=I_n$, and $\Lambda=\mathrm{diag}\{\lambda_i(L)\}$ with $0=\lambda_1(L)\leq\lambda_2(L)\leq \cdots\leq \lambda_n(L)$. 

Using now \eqref{laplcian_decomp} we can rewrite $T(s)$ as
\begin{align}
    T(s)&=\; (I_n+\mathrm{diag}\{g_i(s)\}f(s)L)^{-1}\mathrm{diag}\{g_i(s)\}\nonumber\\
    &=\; (\mathrm{diag}\{g^{-1}_i(s)\}+f(s)L)^{-1}\nonumber\\
    &=\; (\mathrm{diag}\{g^{-1}_i(s)\}+f(s)V\Lambda V^T)^{-1}\nonumber\\
    &=\; V(V^T\mathrm{diag}\{g^{-1}_i(s)\}V+f(s)\Lambda)^{-1}V^T\,.\label{eq_T_eigenform}
\end{align}

As mentioned before, we are interested in characterizing the behavior of the closed-loop system $T(s)$ of \eqref{eq_T_explict} as the connectivity of $L$, i.e. $\lambda_2(L)$, increases. To gain some insight we first consider the following simplified example.
\subsection{Motivating Example: Homogeneous Node Dynamics}

{\ifthenelse{\boolean{color}}{\color{blue}}{}Suppose $g_i(s)$ are homogeneous, i.e., $g_i(s)=g(s)$. Then using \eqref{eq_T_eigenform} one can decompose $T(s)$ as follows
\be
    T(s) = \frac{1}{n}g(s)\one\one^T+V_\perp\dg\lb \frac{1}{g^{-1}(s)+f(s)\lambda_i(L)}\rb_{i=2}^n V_\perp^T\,,\label{eq_T_homo_decomp}
\ee
where the network dynamics decouple into two terms: 1) the dynamics $\frac{1}{n}g(s)\one\one^T$ that is independent of network topology and corresponds to the coherent behavior of the system; 2) the remaining dynamics that are dependent on the network structure via both, the eigenvalues $\lambda_i(L), i=2,\cdots,n$ and the eigenvectors $V_\perp$. 
Notice that $|f(s)\lambda_2(L)|\leq|f(s)\lambda_i(L)|, i=2,\dots,n$, then $\frac{1}{n}g(s)\one\one^T$ is dominant in $T(s)$ as long as $|f(s)\lambda_2(L)|$ (later referred as \emph{effective algbraic connectivity}),  is large enough to make the norm of the second term in \eqref{eq_T_homo_decomp} sufficiently small. Following such observation, we can find two regimes where the coherent dynamics $\frac{1}{n}g(s)\one\one^T$ is dominant:
\begin{enumerate}
    \item (\emph{High network connectivity}) for almost every $s_0\in\compl$, except for the poles of $g(s)$, the following holds:
    \ben
        \lim_{\lambda_2(L)\ra \infty}\lV T(s_0)-\frac{1}{n}g(s_0)\one\one^T\rV=0\,.
    \een
    Furthermore, one can verify that if a compact set $S\subset \compl$ contains neither zeros nor poles of $g(s)$, the following holds:
    \ben
        \lim_{\lambda_2(L)\ra \infty}\sup_{s\in S}\lV T(s)-\frac{1}{n}g(s)\one\one^T\rV=0\,.
    \een
    \item (\emph{High gain in coupling dynamics}) If $s_0$ is a pole of $f(s)$, and $\lambda_2(L)\neq 0$, then
    $$\lim_{s\ra s_0}\lV T(s)-\frac{1}{n}g(s)\one\one^T\rV=0\,.$$
\end{enumerate}

Such convergence results suggest that if 1) the network has high algebraic connectivity, or 2) our point of interest in frequency domain is close to pole of $f(s)$, the response of the entire system is close to one of $\frac{1}{n}g(s)\one\one^T$. We refer $\frac{1}{n}g(s)\one\one^T$ as the coherent dynamics\footnote{We also refer $g(s)$ as the coherent dynamics since transfer matrix of the form $\frac{1}{n}g(s)\one\one^T$ is uniquely determined by its non-zero eigenvalue $g(s)$.} in the sense that in such system, the inputs are aggregated, and all nodes have exactly the same response to the aggregate input.
\emph{Therefore, coherence of the network corresponds, in the frequency domain, to  the property that the network's transfer matrix approximately having a particular rank-one structure, and thus coherence increases as we increase the network connectivity}, as depicted in Fig. \ref{fig_diagram_coherence}.
\begin{figure}[ht]
    \centering
    \includegraphics[width=8cm]{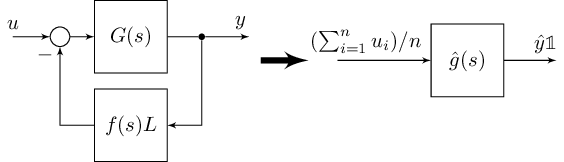}
    \caption{The network coherence can be understood as approximating (left) the closed-loop networks dynamics $T(s)$ by (right) a coherent dynamics governed by an SISO dynamics $\hat{g}(s)$.}
    \label{fig_diagram_coherence}
\end{figure}

The aforementioned analysis can be extended to the case with proportionality assumption, i.e., $g_i(s)=p_ig(s)$ for some $g(s)$ and $p_i>0,i=1,\cdots,n$, where one can still obtain decoupled dynamics through proper coordinate transformation~\cite{Paganini2019tac}. However, it is challenging to characterize the coherent dynamics without the proportionality assumption. Our work precisely aims at understanding the coherent dynamics of non-proportional heterogeneous networks.
\begin{rem}
    In this paper, we write most of our convergence results in the high connectivity regime as the limit of differences in norm when $\lambda_2(L)\ra \infty$ for simplicity. However, one does not require infinitely high connectivity to achieve coherence. These limits suggests, under sufficiently high connectivity, the transfer matrix $T(s)$ is, in some sense, close to coherent dynamics $\frac{1}{n}g(s)\one\one^T$. The precise non-asymptotic result is presented in Lemma \ref{lem_reg_norm_bd}. Moreover, notice that in the other regime around pole of $f(s)$, we only requires the network to be connected.
\end{rem}
}

\subsection{The Goal of This Work}\label{ssec:goal}
In this paper, we would like to show that even when $g_i(s)$ are heterogeneous, similar convergence result still holds. More precisely, we will, in the following sections, discuss the point-wise and uniform convergence of $T(s)$ to a transfer matrix of the form $\frac{1}{n}\bar{g}(s)\one\one^T$, as the effective algebraic connectivity $|f(s)\lambda_2(L)|$ increases. 
However, unlike the homogeneous node dynamics case where the coherent behavior is driven by $\bar g(s)=g(s)$, we will show that the coherent dynamics $\bar{g}(s)$ are given by the harmonic mean of $g_i(s),i=1,\cdots,n$, i.e.,
\be
    \bar{g}(s)=\lp \frac{1}{n}\sum_{i=1}^ng_i^{-1}(s)\rp^{-1}\,.\label{eq_g_bar}
\ee

Such asymptotic analyses under high connectivity serve two main purposes. 
Firstly, using the coherent dynamics $\bar{g}(s)$, one can infer the point-wise convergence of $T(s)$ as $\lambda_2(L)$ increases. In particular, we show that
\begin{enumerate}
    \item For a point that is neither a pole nor a zero of $\bar{g}(s)$, $T(s)$ converges to $\frac{1}{n}\bar{g}(s)\one\one^T$;
    \item Poles of $\bar{g}(s)$ are asymptotically poles of $T(s)$;
    \item Zeros of $\bar{g}(s)$ are asymptotically zeros of $T(s)$.
\end{enumerate}
Secondly, uniform convergence of $T(s)$ explains the coherent behavior/response of a tightly-connected network subject to disturbances. To see the connection, recall the Inverse Laplace Transform~\cite[Theorem 3.20]{Dullerud2013} computes the system time-domain response by integration on the line $\{\sigma+j\omega: \omega\in[-\infty,+\infty]\}$ in frequency domain with a suitable $\sigma$. Then uniform convergence of $T(s)$ on this line would show that time-domain response of the network converges to one of the coherent dynamics $\frac{1}{n}\bar{g}(s)\one\one^T$ as network connectivity increases. However, we will see that such convergence does not hold for most networks through our analysis and subsequent examples, which suggests that \emph{the coherence we analyze here is a frequency-dependent phenomenon}. On that note, one generally resort to weaker convergence results on a line segment $\{\sigma+j\omega: \omega\in[-\omega_c,+\omega_c]\}$ for some $\omega_c>0$, which, once established, justifies the coherent behavior of tightly-connected networks subject to low-frequency disturbances (see Section \ref{ssec:to_time}). Since the set above is a compact subset of $\compl$, we are mostly discussing the uniform convergence results over compact sets. 
{\ifthenelse{\boolean{color}}{\color{blue}}{}Moreover, we also show the convergence in the other regime: $T(s)$ is coherent around the pole of $f(s)$. This suggests that the network responses coherently, when subject to input signals that has its frequency component concentrate around pole of $f(s)$, and we show such an example in power networks in Section \ref{sec:exmples}.}

One additional feature of our analysis is that it can be further applied in settings where the composition of the network is unknown and only distributional information is present. More precisely, we will extend such convergence results by considering a tightly-connected network where node dynamics are given by random transfer functions. As the network size grows, the coherent dynamics $\bar{g}(s)$ converges in probability to a deterministic transfer function. We term such a phenomenon, where a family of uncertain large-scale tightly-connected systems concentrates to a common deterministic system, \emph{dynamics concentration}.

{\ifthenelse{\boolean{color}}{\color{blue}}{} 
\begin{rem}
    In order for the frequency domain analysis provided to be meaningful, it is necessary that both the transfer functions $T(s)$ and $T(s)-\one\one^T\bar g(s)$ are stable. Since our primary focus is on the interpretation of the frequency domain results, we are largely working under the tacit assumption that these transfer functions are stable whenever required. Some simple passivity motivated criteria that ensure stability even as $\lambda_2(L)$ becomes arbitrarily large are given in Theorem~\ref{thm_ptw_conv_pole}. It should also be noted that there exist a range of scalable stability criteria in the literature that can be used to guarantee internal stability of the feedback setup in Figure~\ref{fig_diagram_coherence}. Perhaps the most well known is that if each $g_i(s)$ is \textit{strictly positive real}, and $f(s)$ is \textit{positive real}, then the transfer functions $\bar{g}(s)$ and
    \[
    \begin{bmatrix}
    G(s)\\I
    \end{bmatrix}\left(I+f(s)LG(s)\right)^{-1}\begin{bmatrix}
    f(s)L&I
    \end{bmatrix}
    \]
    are stable (see e.g. \cite{MD95}). Alternative approaches that can be easily adapted to our framework that give criteria that allow for different classes of transfer functions include \cite{LV06,JK10,pm19tcns}.
\end{rem}}

\subsection{Preliminaries}
Before presenting our results, we
first state a few facts and preliminary results that are used in later proofs.
\subsubsection{Basic Results on Vectors and Matrices}
The following results can be found in~\cite{Horn:2012:MA:2422911}.

\begin{itemize}
    \item \emph{Norm Inequalities:} Let $x,y\in \mathbb{R}^n$ and $A,B\in\mathbb{R}^{m\times n}$, we have
    \begin{align}
        &\; \text{2-Norm compatibility:} & \|Ax\|\leq \|A\|\|x\|\,,\label{eq_norm_comp}\\
        &\; \text{2-Norm sub-multiplicativity:} & \|AB\|\leq \|A\|\|B\|\,,\label{eq_norm_sub_multi}\\
        &\; \text{Cauchy-Schwarz:} & |x^Ty|\leq\|x\|\|y\|\label{eq_cs_ineq}\,.
    \end{align}
    \item \emph{Inverse of Block Matrix:} For block matrix $M=\bmt A&B\\ C&D\emt$ with $A,D$ being square matrices, its inverse can be written as
    \be
        M^{-1}=\bmt S^{-1}& -S^{-1}BD^{-1}\\
        -D^{-1}CS^{-1} & D^{-1}+D^{-1}CS^{-1}BD^{-1}\emt\,,\label{eq_blk_mat_inv}
    \ee
    where $S=A-BD^{-1}C$, provided that all relevant inverses exist.
\end{itemize}

\subsubsection{Inequalities for singular values}
We also provide several inequalities for matrix singular values from~\cite[7.3.P16]{Horn:2012:MA:2422911} which will be used in our proofs.
\begin{lem}[Weyl's Inequality]\label{lem_sing_val_ineq}
    Let $A,B$ be square matrices of order $n$, the following inequalities hold:
    \begin{subequations}
    \begin{align}
        \|AB\|&\geq\; \|A\|\sigma_1(B)\label{eq_sv_prod_lb}\,,\\
        \sigma_1(AB)&\leq\; \|A\|\sigma_1(B)\,,\label{eq_sv_prod_ub}\\
        \sigma_1(A+B)&\leq\; \sigma_1(A)+\|B\|\,. \label{eq_weyl_ineq_sv}
    \end{align}
    \end{subequations}
\end{lem}

Lemma \ref{lem_sing_val_ineq} allows us to obtain a useful bound on the spectral norm of $(A+B)^{-1}$. We state it as another Lemma as it will be repeatedly used in sections \ref{sec:ptw_conv} and \ref{sec:unifm_conv}:
\begin{lem}\label{lem_bd_norm_mat_inv}
    Let $A,B$ be square matrices of order $n$. If $\sigma_1(A)\geq \|B\|>0$, then the following inequality holds:
    \ben
        \|(A+B)^{-1}\| \leq\frac{1}{\sigma_1(A)-\|B\|}\,.
    \een
\end{lem}
\begin{proof}
    By \eqref{eq_weyl_ineq_sv}, we have
    \ben
        \sigma_{1}(A)\leq \sigma_{1}(A+B)+\|-B\|\,.
    \een
    Then as long as $\sigma_1(A)\geq \|B\|>0$, it leads to 
    \ben
        \frac{1}{\sigma_1(A+B)}\leq \frac{1}{\sigma_1(A)-\|B\|}\,,
    \een
    and the left-hand side is exactly $\|(A+B)^{-1}\|$.
\end{proof}

\subsubsection{Grounded Laplacian Matrix}\label{app_grd_lap}
For a $n\by n$ Laplacian matrix $L$, we select an index set $I\subset [n]$. Then the grounded Laplacian $\Tilde{L}$ is the principal submatrix of $L$ obtained by removing the rows and columns corresponding to the index set $I$. The following lemma relates the eigenvalues of $\Tilde{L}$ and $L$.
\begin{lem}\label{lem_grd_Lap_eig_bd}
    Given a $n\by n$ symmetric Laplacian matrix $L$, let $\Tilde{L}$ be its grounded Laplacian corresponding to a index set $\mathcal{I}$ with $|\mathcal{I}|=m<n$. Then for the least eigenvalue of $\Tilde{L}$, the following inequality holds:
    \begin{equation*}
        \lambda_1(\Tilde{L})\geq \frac{m}{n}\lambda_2(L)\,.
    \end{equation*}
\end{lem}
\ifthenelse{\boolean{archive}}{The proof is shown in the Appendix.\ref{app_proof_lem_grd_Lap_eig_bd}.}{We refer the readers to~\cite{min2021a} for the proof.} This lower bound shows that for weighted graphs with fixed network size $n$, as $\lambda_2(L)\ra \infty$, we also have $\lambda_1(\Tilde{L})\ra \infty$. This result is used in section \ref{sec:ptw_conv} and \ref{sec:unifm_conv} when we present the convergence analysis regarding zeros of $\bar{g}(s)$.

Now we are ready for the main results of this paper, and we start with the exact characterization of coherent dynamics of tightly-connected networks by proving the point-wise convergence of $T(s)$.

\section{Point-wise Coherence}\label{sec:ptw_conv}
    
In this section, we analyze the strength of network coherence at a single point in the frequency domain. We start with an important lemma revealing how such coherence is related to algebraic connectivity $\lambda_2(L)$ and feedback dynamics $f(s)$. 

\begin{lem}\label{lem_reg_norm_bd}
    Let $T(s)$ and $\bar{g}(s)$ be defined as in \eqref{eq_T_explict} and \eqref{eq_g_bar}, respectively. Suppose that for $s_0\in\compl$ that is not a pole of $f(s)$, we have $$|\bar{g}(s_0)|\leq M_1, \text{and }\max_{1\leq i\leq n}|g_i^{-1}(s_0)|\leq M_2\,,$$ for some $M_1,M_2>0$. Then for large enough $\lambda_2(L)$, the following inequality holds:
    \be
        \lV T(s_0)-\frac{1}{n}\bar{g}(s_0)\one\one^T\rV\leq \frac{\lp M_1M_2+1\rp^2}{|f(s_0)|\lambda_2(L)-M_2-M_1M_2^2}\,.\label{eq_T_norm_bd}
    \ee
\end{lem}

\begin{proof}
    Let $H=V^T\dg\{g_i^{-1}(s_0)\}V+f(s_0)\Lambda$, such that 
    \eqref{eq_T_eigenform} becomes\[
    T(s) = VH^{-1}V^T.
    \]
    Then it is easy to see that
    \begin{align}
        \lV T(s_0)-\frac{1}{n}\bar{g}(s_0)\one\one^T\rV&=\; \|T(s_0)-\bar{g}(s_0)Ve_1e_1^TV^T\|\nonumber\\
        &=\; \lV V\lp H^{-1}-\bar{g}(s_0)e_1e_1^T\rp V^T\rV\nonumber\\
        &=\; \lV H^{-1}-\bar{g}(s_0)e_1e_1^T \rV\,,\label{eq_T_H_norm_equiv}
    \end{align}
    where $e_1$ is the first column of identity matrix $I_n$. The first equality holds by noticing that $\frac{\one}{\sqrt{n}}$ is the first column of $V$, and the last equality comes from the fact that multiplying by a unitary matrix $V$ preserves the spectral norm.
    
    We now write $H$ in block matrix form:
    \begin{align*}
        H&=\;V^T\dg\{g_i^{-1}(s_0)\}V+f(s_0)\Lambda\\
        &= \bmt
            \frac{\one^T}{\sqrt{n}}\\
            V_\perp^T
        \emt \dg\{g_i^{-1}(s_0)\} \bmt
        \frac{\one}{\sqrt{n}}& V_\perp\emt+f(s_0)\Lambda\\
        &= {\small\bmt
        \bar{g}^{-1}(s_0)& \frac{\one^T}{\sqrt{n}}\dg\{g_i^{-1}(s_0)\}V_\perp\\
        V_\perp^T\dg\{g_i^{-1}(s_0)\}\frac{\one}{\sqrt{n}} &V_\perp^T\dg\{g_i^{-1}(s_0)\}V_\perp+f(s_0)\Tilde{\Lambda}
        \emt}\\
        &:= \bmt
        \bar{g}^{-1}(s_0)& h^T_{21}\\
        h_{21} & H_{22}
        \emt\,,
    \end{align*}
    where  $\Tilde{\Lambda}=\dg\{\lambda_2(L),\cdots,\lambda_n(L)\}$, and we use the fact that $\lambda_1(\Lambda)=0$.
    
    Inverting $H$ in its block form as in \eqref{eq_blk_mat_inv}, we have
    \ben
        H^{-1} = \bmt
        a &-ah_{21}^TH_{22}^{-1}\\
        -aH_{22}^{-1}h_{21}& H_{22}^{-1}+aH_{22}^{-1}h_{21}h_{21}^TH_{22}^{-1}
        \emt\,,
    \een
    where $a = \frac{1}{\bar{g}^{-1}(s_0)-h_{21}^TH_{22}^{-1}h_{21}}$.
    
    Notice that $||V_\perp||=1$ and $||\one||=\sqrt{n}$, we have
    \begin{align}
        \|h_{21}\|&=\; \lV V_\perp^T\dg\{g_i^{-1}(s_0)\}\frac{\one}{\sqrt{n}}\rV\nonumber\\
        &\leq\; \|V_\perp\|\|\dg\{g_i^{-1}(s_0)\}\|\frac{\|\one\|}{\sqrt{n}}\leq M_2\,,\label{eq_h12_norm_bd}
    \end{align}
    where \eqref{eq_h12_norm_bd} follows from the norm compatibility \eqref{eq_norm_comp} and that matrix 2-norm is sub-multiplicative \eqref{eq_norm_sub_multi}.
    
    Also, by Lemma \ref{lem_bd_norm_mat_inv}, when $|f(s_0)|\lambda_2(L)>M_2$, the following holds:
    \begin{align}
        \|H_{22}^{-1}\|&=\;\|(f(s_0)\Tilde{\Lambda}+ V_\perp^T\dg\{g_i^{-1}(s_0)\}V_\perp)^{-1}\|\nonumber\\
        &\leq\; \frac{1}{\sigma_1(f(s_0)\Tilde{\Lambda})-\|V_\perp^T\dg\{g_i^{-1}(s_0)\}V_\perp\|}\nonumber\\
        &\leq\; \frac{1}{\sigma_1(f(s_0)\Tilde{\Lambda})-M_2}\leq \frac{1}{|f(s_0)|\lambda_2(L)-M_2} \,.\label{eq_H22_norm_bd}
    \end{align}
    Again \eqref{eq_H22_norm_bd} uses the fact that $\|V_\perp\|=1$, the function $\frac{1}{y-x}$ is decreasing in $y$ and increasing in $x$, and, by our assumption, $\|\dg\{g_i^{-1}(s_0)\}\|=\max_{1\leq i\leq n}|g_i^{-1}(s_0)|\leq M_2$.
    
    Lastly, when $|f(s_0)|\lambda_2(L)>M_2+M_2^2M_1$, a similar reasoning as above, using \eqref{eq_h12_norm_bd} \eqref{eq_H22_norm_bd}, and our assumption $|\bar{g}(s_0)|\leq M_1$, gives
    \begin{align}
        |a|&\leq\; \frac{1}{|\bar{g}^{-1}(s_0)|-|h_{21}^TH_{22}^{-1}h_{21}|}\nonumber\\
        &\leq\; \frac{1}{|\bar{g}^{-1}(s_0)|-\|h_{21}\|^2\|H_{22}^{-1}\|}\nonumber\\
        &\leq\; \frac{1}{\frac{1}{M_1}-\frac{M_2^2}{|f(s_0)|\lambda_2(L)-M_2}}\nonumber\\
        &= \;
        \frac{(|f(s_0)|\lambda_2(L)-M_2)M_1}{|f(s_0)|\lambda_2(L)-M_2-M_1M_2^2}\,,\label{eq_a_norm_bd}
    \end{align}
    where in the second inequality, we used norm compatibility and Cauchy-Schwarz inequality \eqref{eq_cs_ineq} to upper-bound $|h_{21}^TH_{22}h_{21}|$.
    
    Now we bound the norm of $H^{-1}-\bar{g}(s_0)e_1e_1^T$ by the sum of norms of all its blocks:
    \begin{align}
        &\;\|H^{-1}-\bar{g}(s_0)e_1e_1^T\|\nonumber\\
        =&\; \lV \bmt
        a\bar{g}(s_0)h_{21}^TH_{22}^{-1}h_{21} &-ah_{21}^TH_{22}^{-1}\\
        -aH_{22}^{-1}h_{21}& H_{22}^{-1}+aH_{22}^{-1}h_{21}h_{21}^TH_{22}^{-1}
        \emt\rV\nonumber\\
        \leq &\; |a\bar{g}(s_0)h_{21}^TH_{22}^{-1}h_{21}|+2\|aH_{22}^{-1}h_{21}\|\nonumber\\
        &\; +\|H_{22}^{-1}+aH_{22}^{-1}h_{21}h_{21}^TH_{22}^{-1}\|\nonumber\\
        \leq &\; |a|\|H_{22}^{-1}\|(|\bar{g}(s_0)|\|h_{21}\|^2+2\|h_{21}\|+\|h_{21}\|^2\|H_{22}^{-1}\|)\nonumber\\
        &\;\quad\quad +\|H_{22}^{-1}\|\,,\label{eq_Hinv_norm_bd1}
    \end{align}
    Using \eqref{eq_h12_norm_bd}\eqref{eq_H22_norm_bd}\eqref{eq_a_norm_bd}, we can further upper bound \eqref{eq_Hinv_norm_bd1} as
    \begin{align}
        &\;\|H^{-1}-\bar{g}(s_0)e_1e_1^T\|\nonumber\\
        \leq &\;  \frac{M_1^2M_2^2+2M_1M_2+\frac{M_1M_2^2}{|f(s_0)|\lambda_2(L)-M_2}}{|f(s_0)|\lambda_2(L)-M_2-M_1M_2^2}\nonumber\\
        &\;\qquad+\frac{1}{|f(s_0)|\lambda_2(L)-M_2}\nonumber\\
        =&\;\frac{\lp M_1M_2+1\rp^2}{|f(s_0)|\lambda_2(L)-M_2-M_1M_2^2}\,.\label{eq_Hinv_norm_bd2}
    \end{align}
    This bound holds as long as $|f(s_0)|\lambda_2(L)>M_2+M_2^2M_1$. Combining \eqref{eq_T_H_norm_equiv} and \eqref{eq_Hinv_norm_bd2} gives the desired inequality.
\end{proof}

{\ifthenelse{\boolean{color}}{\color{blue}}{}Lemma \ref{lem_reg_norm_bd} provides an upper bound for the incoherence measure we are interested in, namely how far apart the system transfer matrix is, at a particular point in the frequency domain, from being rank-one with coherent direction $\frac{1}{n}\one\one^T$. Notice that this incoherence measure also provide upper bounds for
\begin{align*}
    \max_{ij}|T_{ij}(s_0)-\bar{g}(s_0)|&\leq\; \lV T(s_0)-\frac{1}{n}\bar{g}(s_0)\one\one^T\rV\,\\
    \max_{i,j,k,l}|T_{ij}(s_0)-T_{kl}(s_0)|&\leq\; 2\lV T(s_0)-\frac{1}{n}\bar{g}(s_0)\one\one^T\rV\,,
\end{align*}
thus the transfer function from any input channel to any node output is approximately $\bar{g}(s)$, if the incoherence measure is small.

We make following additional remarks:

Lemma 4 provides a non-asymptotic rate for our incoherence measure
\begin{equation}\label{eq:non-asymp-rate}
    \lV T(s_0)-\frac{1}{n}\bar{g}(s_0)\one\one^T\rV\sim \mathcal{O}\lp \frac{M_1^2M_2^2}{|f(s_0)|\lambda_2(L)}\rp\,.
\end{equation}
A large value of $|f(s_0)|\lambda_2(L)$ is sufficient to have the incoherence measure small, and we term this quantity as \emph{effective algebraic connectivity}. We see that there are two possible ways to achieve such point-wise coherence: Either we increase the network algebraic connectivity $\lambda_2(L)$, by adding edges to the network and increasing edge weights, etc., or we move our point of interest $s_0$ to a pole of $f(s)$. This point-wise coherence via effective connectivity provides the basis of our subsequent analysis. Moreover, such coherence could be achieved by practical networks, and in Section \ref{sec:exmples}, we apply our results to understand the coherent response of power generator networks.}

Secondly, the upper bound is frequency-dependent since it is provided at a single point $s_0$ in the $s$-domain. To see such dependence, notice that $s_0$ near a pole of $f(s)$ has large effective algebraic connectivity, hence the system is more coherent around poles of $f(s)$; On the contrary, $s_0$ near a pole of $\bar{g}(s)$ requires large $M_1$ for the condition of Lemma \ref{lem_reg_norm_bd} to hold, and readers can check that $s_0$ near a zero of $\bar{g}(s)$ requires large $M_2$, therefore for at these points, it is more difficult for us to upper bound the incoherence measure by Lemma \ref{lem_reg_norm_bd}. Such dependence makes it challenging to understand the network coherence uniformly in the entire frequency domain.
    
Last but not least, Although Lemma \ref{lem_reg_norm_bd} provides a sufficient condition for the network coherence to emerge, i.e. the increasing effective algebraic connectivity, it is still unknown whether such a condition is necessary. In other words, we do not know whether low effective algebraic connectivity means some kind of incoherence. This problem seems trivial for the extreme case: if $|f(s_0)|=0$ or $L=0$, the feedback loop vanishes, and every node responses independently, but certainly not otherwise.
    
When the condition in Lemma \ref{lem_reg_norm_bd} is satisfied, the system is asymptotically coherent, i.e. $T(s_0)$ converges to $\frac{1}{n}\bar{g}(s_0)$ as the effective algebraic connectivity $|f(s_0)|\lambda_2(L)$ increases. As we mentioned above, we can achieve this by increasing either $\lambda_2(L)$ or $|f(s_0)|$, provided that the other value is fixed and non-zero. Subsection \ref{ssec:gen_p_of_f} considers the former and  Subsection \ref{ssec:poles_of_f} the latter.
    
Before presenting with the results, we define
\begin{dfn}
    For transfer function $g(s)$ and $s_0\in \compl$, $s_0$ is a generic point of $g(s)$ if $s_0$ is neither a pole nor a zero of $g(s)$.
\end{dfn}
As we have seen through the discussions above, we always require some generic point assumptions for either $\bar{g}(s)$, $f(s)$, or both.  Those points are of the most interest in this paper but we will provide some results for the cases where the generic assumption fails.
    
\subsection{Convergence at Generic Points of $f(s)$}\label{ssec:gen_p_of_f}
In this section we keep $s_0$ fixed and present the point-wise convergence result of $T(s_0)$ as $\lambda_2(L)$ increases. This requires $s_0$ to be a generic point of $f(s)$.
    
Notice that for any $s_0$ that is also a generic point of $\bar{g}(s)$, we can always find such $M_1,M_2>0$ and large enough $\lambda_2(L)$ for the upper bound in \eqref{eq_T_norm_bd} to hold. Furthermore, given fixed $M_1$ and $M_2$, one can let the upper bound be arbitrarily small by increasing $\lambda_2(L)$, which leads to the point-wise convergence of $T(s_0)$, as stated in the following theorem.
\begin{thm}\label{thm_ptw_conv_reg}
    Let $T(s)$ and $\bar{g}(s)$ be defined as in \eqref{eq_T_explict} and \eqref{eq_g_bar}, respectively. If $s_0\in\compl$ is a generic point of both $\bar{g}(s)$ and $f(s)$, then
    \ben
        \lim_{\lambda_2(L)\ra +\infty} \lV T(s_0)-\frac{1}{n}\bar{g}(s_0)\one\one^T\rV=0\,.
    \een
\end{thm}
\begin{proof}
    Since $s_0$ is not a pole of $\bar{g}(s)$, $|\bar{g}(s_0)|$ is trivially upper bounded by some $M_1>0$. Also, it is easy to see that $s_0$ is not a zero of $\bar{g}(s)$ if and only if $s_0$ is not a zero of any $g_i(s)$. Then $\max_{1\leq i\leq n}|\bar{g}^{-1}(s_0)|$ is upper bounded by some $M_2>0$. Therefore the conditions of Lemma \ref{lem_reg_norm_bd} are satisfied. We finish the proof by taking $\lambda_2(L)\ra +\infty$ on both sides of \eqref{eq_T_norm_bd}.
\end{proof}

Theorem \ref{thm_ptw_conv_reg} establishes the emergence of coherence at generic points of $\bar{g}(s)$. This forms the basis of our analysis, yet requires such $s_0$ satisfying generic conditions. A more careful analysis shows that, as $\lambda_2(L)\ra +\infty$, the pole of $\bar{g}(s)$ is asymptotically a pole of $T(s)$, and the zero of $\bar{g}(s)$ is asymptotically a zero of $T(s)$, as stated in the following two theorems. \ifthenelse{\boolean{archive}}{}{We refer the readers to \cite{min2021a} for the proofs for Theorem \ref{thm_ptw_conv_pole} and \ref{thm_ptw_conv_zero}.}
\begin{thm}\label{thm_ptw_conv_pole}
    Let $T(s)$ and $\bar{g}(s)$ be defined as in \eqref{eq_T_explict} and \eqref{eq_g_bar}, respectively. If $s_0\in\compl$ is a pole of $\bar{g}(s)$ and a generic point of $f(s)$, then
    \ben
        \lim_{\lambda_2(L)\ra +\infty} \lV T(s_0)\rV=+\infty\,.
    \een
\end{thm} 
\ifthenelse{\boolean{archive}}{
\begin{proof}
    Similarly to the proof of Lemma \ref{lem_reg_norm_bd}, we define $H=V^T\dg\{g_i^{-1}(s_0)\}V+f(s_0)\Lambda$ and now we need to show that $\|T(s_0)\|=\|H^{-1}\|$ grows unbounded as $\lambda_2(L)\ra +\infty$. 
    
    Write $H$ in block matrix form:
    \begin{align*}
         H&= {\small\bmt
        \bar{g}^{-1}(s_0)& \frac{\one^T}{\sqrt{n}}\dg\{g_i^{-1}(s_0)\}V_\perp\\
        V_\perp^T\dg\{g_i^{-1}(s_0)\}\frac{\one}{\sqrt{n}} &V_\perp^T\dg\{g_i^{-1}(s_0)\}V_\perp+f(s_0)\Tilde{\Lambda}
        \emt}\\
        &:= \bmt
        0& h^T_{21}\\
        h_{21} & H_{22}
        \emt\,,
    \end{align*}
    by noticing that $\bar{g}^{-1}(s_0)=0$ because $s_0$ is a pole of $\bar{g}(s)$.
    
    Inverting $H$ in its block form gives
    \begin{align}
        H^{-1} &=\; \bmt
        a &-ah_{21}^TH_{22}^{-1}\\
        -aH_{22}^{-1}h_{21}& H_{22}^{-1}+aH_{22}^{-1}h_{21}h_{21}^TH_{22}^{-1}
        \emt\nonumber\\
        &=\;
        a\bmt 1\\
        -H_{22}^{-1}h_{21}\emt\bmt1 &-h^T_{21}H_{22}^{-1}\emt+\bmt 0&0\\
        0 & H_{22}^{-1}\emt\,, \label{eq_H_inv_pole}
    \end{align}
    where $a$ now is given by $a = -\frac{1}{h_{21}^TH_{22}^{-1}h_{21}}$.
    
    Then from \eqref{eq_H_inv_pole}, when $\lambda_2(L)$ is large enough, we can lower bound $\|H^{-1}\|$ by
    \begin{align}
        \|H^{-1}\|&\geq \lV a\bmt 1\\
        -H_{22}^{-1}h_{21}\emt\bmt 1\\
        -H_{22}^{-1}h_{21}\emt^T\rV - \lV\bmt 0&0\\
        0 & H_{22}^{-1}\emt \rV\nonumber\\
        &= \frac{1}{|h_{21}^TH_{22}^{-1}h_{21}|} \lV \bmt 1\\
        -H_{22}^{-1}h_{21}\emt\rV^2-\|H_{22}^{-1}\|\nonumber\\
        &\geq \frac{1}{|h_{21}^TH_{22}^{-1}h_{21}|}-\|H_{22}^{-1}\|\nonumber\\
        &\geq \frac{1}{\|h_{21}\|^2\|H_{22}^{-1}\|}-\|H_{22}^{-1}\|\,,\label{eq_H_lower_bd_pole}
    \end{align}
    where in the second inequality, we simply use the fact that the norm of a vector is lower bounded by its first entry.
    
    Because $s_0$ is a pole of $\bar{g}(s)$, it cannot be a zero of any $g_i(s)$; otherwise this would lead to the contradiction $\bar{g}(s_0)=0$. Therefore, $\max_{1\leq i\leq n}|g_i^{-1}(s)|$ is upper bounded by some $M>0$. Similarly to \eqref{eq_h12_norm_bd} and \eqref{eq_H22_norm_bd}, we have $\|h_{21}\|\leq M$ and $\|H_{22}^{-1}\|\leq \frac{1}{|f(s_0)|\lambda_2(L)-M}$. Then \eqref{eq_H_lower_bd_pole} can be lower bounded by
    \begin{align*}
        \|H^{-1}\|&\geq \frac{1}{\|h_{21}\|^2\|H_{22}^{-1}\|}-\|H_{22}^{-1}\|\\
        &\geq \frac{1}{\frac{M^2}{|f(s_0)|\lambda_2(L)-M}}-\frac{1}{|f(s_0)|\lambda_2(L)-M}\\
        &= \frac{(|f(s_0)|\lambda_2(L)-M)^2-M^2}{M^2(|f(s_0)|\lambda_2(L)-M)}\,.
    \end{align*}
    This lower bound holds when $|f(s_0)|\lambda_2(L)\geq M$, and it grows unbounded as $\lambda_2(L)\ra +\infty$, which finishes the proof.
\end{proof}}{}

{
\begin{rem}\label{rem_pole}
    Theorem \ref{thm_ptw_conv_pole} does not suggest whether the network is asymptotically coherent at poles of $\bar{g}(s)$. Our incoherence measure $\lV T(s_0)-\frac{1}{n}\bar{g}(s_0)\one\one^T\rV$ is undefined at such poles. Alternatively, for $s_0$ the pole of $\bar{g}(s)$, one can prove that when $\tilde{\Lambda}/\tilde{\lambda_2(L)}\ra \Lambda_{\mathrm{lim}}$ as $\lambda_2(L)\ra +\infty$, we have the limit $\lV \frac{T(s_0)}{\|T(s_0)\|}-\frac{1}{n}\gamma(\Lambda_\mathrm{lim})\one\one^T\rV\ra 0$, for some $\gamma(\Lambda_\mathrm{lim})\in\compl$ determined by $\Lambda_\mathrm{lim}$ with $|\gamma(\Lambda_\mathrm{lim})|=1$. We leave the formal statement to Appendix.\ref{app_lim_dir_pole}. This suggests that $T(s_0)$ has the desired rank-one structure for coherence. While the normalized transfer matrix is not discussed in this paper due to the space constraints, such formulation is better for understanding the network coherence at the poles of $\bar{g}(s)$.
\end{rem}}
Next, the convergence result regarding the zeros of $\bar{g}(s)$ is stated as
\begin{thm}\label{thm_ptw_conv_zero}
    Let $T(s)$ and $\bar{g}(s)$ be defined as in \eqref{eq_T_explict} and \eqref{eq_g_bar}, respectively. If $s_0\in\compl$ is a zero of $\bar{g}(s)$ and a generic point of $f(s)$, then
    \ben
        \lim_{\lambda_2(L)\ra +\infty} \lV T(s_0)\rV=0\,.
    \een
\end{thm}
\ifthenelse{\boolean{archive}}{\begin{proof}
    Since $s_0$ is a zero of $\bar{g}(s)$, it is the zero of at least one $g_i(s)$. Without loss of generality, suppose $g_i(s_0)=0$ for $1\leq i\leq m$ and $g_i(s_0)\neq 0$ for $m+1\leq i\leq n$.
    
    If $m=n$, then $T(s_0)=0$. We only consider the non-trivial case when $m<n$. The transfer matrix is now given by
    \begin{align}
        T(s_0)&=\;(I_n+G(s_0)f(s_0)L)^{-1}G(s_0)\nonumber\\
        &=\;\bmt I_{m}&0_{m\times (n-m)}\\
        0_{(n-m)\times m}& I_{n-m}+\tilde{G}(s_0)f(s_0)\tilde{L}\emt^{-1}G(s_0)\nonumber\\
        &=\;\bmt
            0_{m\by m}& 0_{m\by (n-m)}\\
            0_{(n-m)\by m} & (I_{n-m}+\Tilde{G}(s_0)f(s_0)\Tilde{L})^{-1}\Tilde{G}(s_0)
        \emt\,,\label{eq_T_zero}
    \end{align}

    where $\Tilde{G}(s)=\dg\{g_{m+1}(s),\cdots,g_{n}(s)\}$ and $\Tilde{L}$ is the \emph{grounded Laplacian} of $L$ by removing the first $m$ rows and columns.
    
    By Lemma \ref{lem_grd_Lap_eig_bd}, when $\lambda_1(\Tilde{L})$ is large enough, we have
    \begin{align*}
        \|T(s_0)\|&=\;\|(I_{n-m}+\Tilde{G}(s_0)f(s_0)\Tilde{L})^{-1}\Tilde{G}(s_0)\|\\
        &=\; \|(\Tilde{G}^{-1}(s_0)+f(s_0)\Tilde{L})^{-1}\|\\
        &\leq\; \frac{1}{\sigma_1(f(s_0)\Tilde{L})-\|\Tilde{G}^{-1}(s_0)\|}\\
        &\leq\; \frac{1}{|f(s_0)|\lambda_1(\Tilde{L})-\|\Tilde{G}^{-1}(s_0)\|}\,.
    \end{align*}
    Since $g_i(s_0)\neq 0$ for $m+1\leq i\leq n$, $\max_{m+1\leq i\leq n}|g_i^{-1}(s_0)|$ is upper bounded by some $M>0$. Then we have
    \be
        \|T(s_0)\|\leq \frac{1}{|f(s_0)|\lambda_1(\Tilde{L})-M}\,.\label{eq_T_bd_zero}
    \ee
    By Lemma \ref{lem_grd_Lap_eig_bd}, we know that $\lambda_1(\Tilde{L})\ra +\infty$ as $\lambda_2(L)\ra \infty$, then
    \ben
        \lim_{\lambda_2(L)\ra +\infty}\frac{1}{|f(s_0)|\lambda_1(\Tilde{L})-M}=0\,.
    \een
    We finishes the proof by taking $\lambda_2(L)\ra +\infty$ on both sides of \eqref{eq_T_bd_zero}
\end{proof}}{}

{
\begin{rem}\label{rem_zero}
    The limit in Theorem \ref{thm_ptw_conv_zero} can still be written as $\lim_{\lambda_2(L)\ra +\infty}\|T(s_0)-\frac{1}{n}\bar{g}(s_0)\one\one^T\|=0$, because $s_0$ is a zero of $\bar{g}(s)$. However, we here emphasize the fact that the system is not coherent  at $s_0$ under normalization because $T(s_0)/\|T(s_0)\|$ does not converge to $\frac{1}{n}\gamma\one\one^T$ for any $\gamma\in\compl$. 
\end{rem}
    
So far, we have shown point-wise convergence of $T(s)$ towards the transfer function $\frac{1}{n}\bar{g}(s)\one\one^T$, from which we assess how network coherence emerges as connectivity increases. In Remark \ref{rem_pole} and \ref{rem_zero} we see that the incoherence measure $\lV T(s_0)-\frac{1}{n}\bar{g}(s_0)\one\one^T\rV$ is insufficient for understanding the asymptotic behavior at zeros or poles of $\bar{g}(s)$, and that the alternative measure $\lV\frac{T(s_0)}{\|T(s_0)\|}-\frac{1}{n}\gamma\one\one^T\rV$ is a good complement for such purpose.\footnote{As $\lambda_2(L)$ increases, for pole of $\bar{g}(s)$, the latter measure converges to $0$ given suitable conditions but not for the former; for zero of $\bar{g}(s)$, the opposite result holds; for generic point of $\bar{g}(s)$, both incoherence measures converge to $0$.} The latter measure, $\lV\frac{T(s_0)}{\|T(s_0)\|}-\frac{1}{n}\gamma\one\one^T\rV$, focuses more on the relative scale of eigenvalues of $T(s)$. In this paper, we mostly use the former, $\lV T(s_0)-\frac{1}{n}\bar{g}(s_0)\one\one^T\rV$, and particularly when presenting the uniform convergence results; because we are interested in connecting these results to the network time-domain response.}
    
\subsection{Convergence Regarding Poles of $f(s)$}\label{ssec:poles_of_f}
As mentioned before, when $s_0$ is a pole of $f(s)$, it is a singularity of $T(s)$. Under certain conditions, one can observe that high-gain in $f(s)$ plays a role similar to $\lambda_2(L)$.
The result uses Lemma \ref{lem_reg_norm_bd} and is stated as follows.
\begin{thm}\label{thm_ptw_singular_f_pole}
    Let $T(s)$ and $\bar{g}(s)$ be defined as in \eqref{eq_T_explict} and \eqref{eq_g_bar}, respectively. Suppose $\lambda_2(L)>0$. If $s_0\in\compl$ is a generic point of $\bar{g}(s)$ and $s_0$ is a pole of $f(s)$, then
    \ben
        \lim_{s\ra s_0} \lV T(s)-\frac{1}{n}\bar{g}(s)\one\one^T\rV=0\,.
    \een
\end{thm}
\begin{proof}
    Since $s_0$ is neither a zero nor a pole of $\bar{g}(s)$, $\exists \delta>0$ such that $\forall s\in U(s_0,\delta)=\{s:|s-s_0|<\delta\}$, we have $|\bar{g}^{-1}(s)|\leq M_1$ and $\max_{1\leq i\leq n}|g_i^{-1}(s)|\leq M_2$ for some $M_1,M_2>0$.
    
    By Lemma \ref{lem_reg_norm_bd}, $\forall s\in U(s_0,\delta)$, the following holds
    \ben
        \lV T(s)-\frac{1}{n}\bar{g}(s)\one\one^T\rV\leq 
        \frac{\lp M_1M_2+1\rp^2}{|f(s)|\lambda_2(L)-M_2-M_1M_2^2}\,.
    \een
    Taking $s\ra s_0$ on both side, notice that $\lim_{s\ra s_0}|f(s)|=+\infty$, the limit of right-hand side is 0. 
\end{proof}
In other words, at pole of $f(s)$, the network effect is infinitely amplified. The effective algebraic connectivity $|f(s)|\lambda_2(L)$ grows unbounded as $s$ approaching the pole of $f(s)$. As a result, the frequency response of $T(s)$ is exactly the one of $\frac{1}{n}\bar{g}(s)\one\one^T$. Hence network coherence naturally arises around poles of $f(s)$.

\section{Uniform Coherence}\label{sec:unifm_conv}

We now leverage the point-wise convergence results of Section \ref{sec:ptw_conv} to characterize conditions for uniform convergence. This will allow us to connect our analysis with time domain implications, as discussed in \ref{ssec:goal}.
    
We start by showing uniform convergence of $T(s)$ over compact regions that do not contain any zero or pole of $\bar{g}(s)$. While the uniform convergence does not hold over regions containing such zeros or poles in general, we prove that in some special cases, the uniform convergence around zeros of $\bar{g}(s)$ does hold. Finally, we provide a sufficient condition for uniform convergence of $T(s)$ on the right-half plane, which implies the system converges in $\mathcal{H}_\infty$ norm. 

\subsection{Uniform Convergence Around Generic Points of $\bar g(s)$}
Again, similarly to the point-wise convergence, we discuss uniform convergence of $T(s)$ over set $S$ that satisfies the following assumption
\begin{asmp}
    $S\subset\compl$ satisfies $\sup_{s\in S}|f(s)|<\infty$ and $\inf_{s\in S}|f(s)|>0$.\label{asmp_unifm}
\end{asmp}
Such an assumption guarantees all points in the closure of $S$ are generic points of $f(s)$. This property prevents any sequence of points in $S$ that asymptotically eliminates or amplifies the network effect on the boundary of $S$. In subsequent sections, we denote $F_h:=\sup_{s\in S}|f(s)|$ and $F_l:=\inf_{s\in S}|f(s)|>0$.  
        
Recall that in Section \ref{sec:ptw_conv}, point-wise convergence is proved by choosing $M_1,M_2>0$  such that the conditions in Lemma \ref{lem_reg_norm_bd} are satisfied at a particular point $s_0$. Then, finding universal $M_1,M_2>0$ that work for every $s_0$ in a set $S\subset\compl$ suffices to show uniform convergence over $S$. Such a process is straightforward if $S$ is compact: 
\begin{thm}\label{thm_unifm_conv_reg_compact}
    Let $T(s)$ and $\bar{g}(s)$ be defined as in \eqref{eq_T_explict} and \eqref{eq_g_bar}, respectively. Then given a compact set $S\subset \compl$, if $S$ satisfies Assumption \ref{asmp_unifm} and does not contain any zero or pole of $\bar{g}(s)$,  we have
    \ben
        \lim_{\lambda_2(L)\ra +\infty}\sup_{s\in S}\lV T(s)-\frac{1}{n}\bar{g}(s)\one\one^T\rV=0\,.
    \een
\end{thm}
\begin{proof}
    On the one hand, since $S$ does not contain any pole of $\bar{g}(s)$, $\bar{g}(s)$ is continuous on the compact set $S$, and hence bounded~\cite[Theorem 4.15]{Rudin1964}. On the other hand, because $S$ does not contain any zero of $\bar{g}(s)$, every $g_i^{-1}(s)$ must be continuous on $S$, and hence bounded as well. It follwos that $\max_{1\leq i\leq n}|g_i^{-1}(s)|$ is bounded on $S$, and the conditions of Lemma \ref{lem_reg_norm_bd} are satisfied for all $s\in S$ with a uniform choice of $M_1$ and $M_2$. By \eqref{eq_T_norm_bd}, we have
    \ben
        \sup_{s\in S}\lV T(s)-\frac{1}{n}\bar{g}(s)\one\one^T\rV\leq 
        \frac{\lp M_1M_2+1\rp^2}{F_l\lambda_2(L)-M_2-M_1M_2^2}\,,
    \een
    where $F_{l}=\inf_{s\in S}|f(s)|$.
    We finish the proof by taking $\lambda_2(L)\ra +\infty$ on both sides.   
\end{proof}
As we already discussed in Remark \ref{rem_pole}, if $S$ contains poles of $\bar{g}(s)$, $\sup_{s\in S}\lV T(s)-\frac{1}{n}\bar{g}(s)\one\one^T\rV$ is not a good incoherence measure as it is undefined. The rest of the section mainly discusses the uniform convergence result around zeros of $\bar{g}(s)$.
\subsection{Uniform Convergence Around Zeros of $\bar{g}(s)$}
We first define the notion of \emph{Nodal Multiplicity} of a point in complex plane w.r.t. a given network.
\begin{dfn}
    Given $\{g_i(s),i\in[n]\}$, the Nodal Multiplicity of $s_0\in\compl$ is defined as
    \ben
        \mathcal{N}(s_0):=|\{i\in[n]:g_i(s_0)=0\}|\,,
    \een
    where $|\cdot|$ denotes the set cardinality.
\end{dfn}
By definition, any zero of $\bar{g}(s)$ must have positive nodal multiplicity. Our finding is that zeros with nodal multiplicity exactly $1$ have a special property, which is shown in the following Lemma.
\begin{lem}\label{lem_unifm_conv_zero_neighbor}
    Let $T(s),\bar{g}(s)$ be defined as in \eqref{eq_T_explict} and \eqref{eq_g_bar}, respectively. If $s_0\in\compl$ is a zero of $\bar{g}(s)$ with nodal multiplicity $\mathcal{N}(s_0)=1$, and we assume that $\exists \delta_0$ such that Assumption \ref{asmp_unifm} holds for $U(s_0,\delta_0):=\{s\in\compl: |s-s_0|<\delta_0\}$. Then $\forall \epsilon>0$, $\exists\delta<\delta_0,\lambda>0$ such that whenever $L$ satisfies $\lambda_2(L)\geq \lambda$, we have
    \ben
        \sup_{s\in U(s_0,\delta)}\lV T(s)-\frac{1}{n}\bar{g}(s)\one\one^T\rV<\epsilon\,.
    \een
\end{lem}

The proof is shown in Appendix \ref{app_proof_lem_unifm_conv_zero_neighbor}. Notice that for given $\epsilon>0$, the $\epsilon$ bound is valid for any $\lambda_2(L)\geq \lambda$, therefore we can prove uniform convergence over compact regions that only contain zeros of $\bar{g}(s)$ with nodal multiplicity $1$.
\begin{thm}[Uniform convergence around points with $\mathcal N(s_0)=1$]\label{thm_unifm_conv_compact}
    Let $T(s),\bar{g}(s)$ be defined as in \eqref{eq_T_explict} and \eqref{eq_g_bar}, respectively. For a compact set $S\subset \compl$ satisfying Assumption \ref{asmp_unifm}, if $S$ does not contain any pole of $\bar{g}(s)$, and $\mathcal{N}(s)\leq 1,\forall s\in S$, then we have
    \ben
        \lim_{\lambda_2(L)\ra +\infty}\sup_{s\in S}\lV T(s)-\frac{1}{n}\bar{g}(s)\one\one^T\rV=0\,.
    \een
\end{thm}
\begin{proof}
    Firstly, let $\{s_k,1\leq k\leq m\}$ be the set of all the zeros of $\bar{g}(s)$ within $S$. Then $\mathcal{N}(s_k)=1,\ 1\leq k\leq m$, and by Lemma \ref{lem_unifm_conv_zero_neighbor}$, \forall \epsilon>0$ and every $1\leq k\leq m$, $\exists\, \delta_{s_k},\lambda_{s_k}>0$ such that $\forall L$ satisfying $\lambda_2(L)\geq \lambda_{s_k}$, the following holds:
    \ben
        \sup_{s\in U(s_k,\delta_{s_k})}\lV T(s)-\frac{1}{n}\bar{g}(s)\one\one^T\rV<\epsilon\,.
    \een
    Let $\hat{S} :=S\setminus{\lp\bigcup_{k=1}^mU(s_k,\delta_{s_k})\rp}$, then we know that $\hat{S}$ is a compact set that does not contain any pole or zero of $\bar{g}$. By Theorem \ref{thm_unifm_conv_reg_compact}, $\exists\hat{\lambda}$ such that
    \ben
        \sup_{s\in\hat{S}}\lV T(s)-\frac{1}{n}\bar{g}(s)\one\one^T\rV<\epsilon\,.
    \een
    Let $\lambda =\max\lb\hat{\lambda},\lambda_{s_1},\cdots,\lambda_{s_m}\rb$, then $\forall L$ satisfying $\lambda_2(L)\geq \lambda$, we have:
    \begin{align*}
        &\;\sup_{s\in S}\lV T(s)-\frac{1}{n}\bar{g}(s)\one\one^T\rV\\
        = &\; \max\lb\sup_{s\in\hat{S}}\lV T(s)-\frac{1}{n}\bar{g}(s)\one\one^T\rV\right. ,\\
        &\;\quad \left.\sup_{s\in \bigcup_{k=1}^mU(s_k,\delta_{s_k})}\lV T(s)-\frac{1}{n}\bar{g}(s)\one\one^T\rV\rb<\epsilon\,,
    \end{align*}
    which proves the limit.
\end{proof}
    
For zeros with nodal multiplicity strictly larger than $1$, the analysis is rather complicated. We first look once again at the homogeneous node dynamics setting of Section \ref{sec:prem} to provide some insight. 
  
\begin{example}
    Consider again a homogeneous network with node dynamics $g(s)$ and $f(s)=1$, where the transfer matrix is given by
    \ben
        T(s) = \frac{1}{n}g(s)\one\one^T+V_\perp\dg\lb \frac{1}{g^{-1}(s)+\lambda_i(L)}\rb_{i=2}^n V_\perp^T\,.
    \een
    The poles of $T(s)$ include 1) the poles of $g(s)$, and 2) any point $s_0$ that satisfies $g^{-1}(s_0)+\lambda_i(L)=0$ for a particular $i$. Notice that if $\lambda_i(L)$ is large, every solution to $g^{-1}(s_0)+\lambda_i(L)=0$ is close to one of the zeros of $g(s)$. As we increase $\lambda_2(L)$, which effectively increases every $\lambda_i(L),2\leq i\leq n$, one can check that at most $n-1$ poles asymptotically approach each zero of $g(s)$, provided that $\lambda_i(L)$ are distinct. As a result, uniform convergence around any zero of $g(s)$ cannot be obtained due to the presence of poles of $T(s)$ close to them.
\end{example}
    
Such observation also seems to hold in general for networks with heterogeneous node dynamics $g_i(s)$. That is, if a zero of $\bar{g}(s)$ is a zero with nodal multiplicity strictly larger than $1$, then we expect it to ``attract" poles of $T(s)$. But it is difficult to formally prove it since we cannot exactly locate the poles of $T(s)$ in the absence of homogeneity.\footnote{We can still exactly locate the poles of $T(s)$ when proportionality is assumed, i.e. $g_i(s)=f_ig(s),i\in[n]$ for some $f_i>0$ and rational transfer function $g(s)$. Such a case can be regarded as the homogeneous case by considering a scaled version of $L$.} Surprisingly, there are certain cases where we can still quantify the effect of those poles of $T(s)$ approaching a zero of $\bar g(s)$. This essentially disproves the uniform convergence for such cases.
\begin{thm}[Uniform Convergence Failure]\label{thm_unifm_conv_fail}
    Let $T(s),\bar{g}(s)$ be defined as in \eqref{eq_T_explict} and \eqref{eq_g_bar}, respectively. Let $f(s)=1$. Suppose $z\in\mathbb{R}$ is a real zero of all $g_i(s),i\in [n]$ with multiplicity 1, i.e. $\forall i\in [n]$, $g_i(z)=0, \lim_{s\ra z}\frac{g_i(s)}{s-z}\neq 0$. Then for any set $S$ containing $z$ in its interior, $\exists \lambda,M>0$ such that, for all Laplacian matrices $L$ satisfying $\lambda_2(L)\geq \lambda$, we have
    \ben
        \sup_{s\in S}\lV T(s)-\frac{1}{n}\bar{g}(s)\one\one^T\rV\geq M\,.
    \een
\end{thm}
\ifthenelse{\boolean{archive}}{The proof is shown in Appendix \ref{app_proof_thm_unifm_conv_fail}}{The proof considers the first-order Taylor approximation of $G(s)$ around the zero $z$. Due to the space constraint, we refer interested readers to~\cite{min2021a} for the proof}. Although Theorem \ref{thm_unifm_conv_fail} only disproves the uniform convergence around one particular type of zero of $\bar{g}(s)$, namely, such zero must have nodal multiplicity $n$ and it must be a real zero with multiplicity 1 for all $g_i(s)$, we believe a similar result holds for any zero that is ``shared" by multiple $g_i(s)$. However, a complete proof is left for future research. 
    
We now provide another point of view of this phenomenon. Suppose in a network of size 2, $z_1$ is exclusively zero of $g_1(s)$ and $z_2$ exclusively zero of $g_2(s)$. If $z_1$ and $z_2$ are close enough, there must be $p$ a pole of $\bar{g}(s)$ in the small neighborhood of $z_1$ or $z_2$. To be more clear, see the following example.
\begin{example}
    Let $g_1(s)=\frac{s+a}{s^2},g_2(s)=\frac{s+a+\epsilon}{s^2}$, then $z_1=-a$ and $z_2=-a-\epsilon$ are the zeros respectively. The coherent dynamics is given by
    \ben
        \bar{g}(s)=\frac{2}{g_1^{-1}(s)+g_2^{-1}(s)}=\frac{(s+a)(s+a+\epsilon)}{2s^2(s+a+\epsilon/2)}\,.
    \een
    $\bar{g}(s)$ has a pole $p=-a-\epsilon/2$ that is in both $\epsilon/2$-neighborhoods of $z_1$ and $z_2$.
\end{example}
By Theorem \ref{thm_ptw_conv_pole}, we know that $p$ is asymptotically a pole of $T(s)$, in other words, there is a pole of $T(s)$ approaching $p$, as the network connectivity increases. Moreover, $z_1$ and $z_2$ being close enough suggests that $p$ is close to $z_1$ and $z_2$, as we see in the example. Consequently, two zeros $z_1,z_2$ being close introduces a pole of $T(s)$ asymptotically approaching a small neighborhood of $z_1,z_2$. Consider the limit case where the two zeros collapse into a shared zero of $g_1(s),g_2(s)$, we should expect a pole of $T(s)$ approaching this shared zero.
    
A similar argument can be made for $m$ zeros of different nodes being close to each other, introducing $m-1$ poles of $\bar{g}(s)$ in the small neighborhood that asymptotically attract poles of $T(s)$. This is by no means a rigorous proof of how uniform convergence fails around a zero ``shared" by multiple $g_i(s)$, but rather a discussion providing  intuition behind such behavior.
    
At this point, we have proved uniform convergence of $T(s)$ on a compact set $S$ that does not include 1) zeros of $\bar{g}(s)$ with Nodal Multiplicity larger than $1$, or 2) poles of $\bar{g}(s)$.
    
In particular, we find that zero with Nodal Multiplicity larger than $1$, i.e. it is "shared" by multiple $g_i(s)$, attracts pole of $T(s)$ as network connectivity increases, which suggests that uniform convergence of $T(s)$ fails around such point. Although we only provide the proof for special cases as in Theorem \ref{thm_unifm_conv_fail}, we conjecture such a statement is true in general and we left more careful analysis for future research. 
    
\subsection{Uniform Convergence on Right-Half Complex Plane}\label{ssec:rhp}
Aside from uniform convergence on compact sets, uniform convergence over the closed right-half plane $\{s:Re(s)\geq 0\}$ is of great interest as well. If we were to establish uniform convergence over the right-half plane for a certain $T(s)$, then given $\bar{g}(s)$ to be stable, the convergence in $\mathcal{H}_\infty$-norm of $T(s)$ towards $\frac{1}{n}\bar{g}(s)\one\one^T$ could be guaranteed, i.e., $T(s)$ converges to $\frac{1}{n}\bar{g}(s)\one\one^T$ as a system. One trivial consequence is that we can infer the stability of $T(s)$ with a large enough $\lambda_2(L)$ by the stability of $\bar{g}(s)$. Furthermore, given any $\mathcal{L}_2$ input signal, we can make the $\mathcal{L}_2$ difference between output responses of $T(s)$ and $\frac{1}{n}\bar{g}(s)\one\one^T$ arbitrarily small by increasing the network connectivity. 
    
Unfortunately, for most networks, we encounter with the same issue we have seen when dealing with zeros of $\bar{g}(s)$: When $g_i(s)$ is strictly proper, $g_i(s)\ra 0$ as $|s|\ra +\infty$, thus, $\infty$ can be viewed as a zero of $g_i(s)$ by regarding $g_i(s)$ as functions defined on extended complex plane $\compl \cup \{\infty\}$. Then for networks that include more than one node whose transfer functions are strictly proper, there will be poles of $T(s)$ approaching $\{\infty\}$ as $\lambda_2(L)$ increases. Notice that those poles could approach $\{\infty\}$ either from the left-half or right-half plane. Apparently, the uniform convergence on the right-half plane will not hold if the latter happens, but even when the former happens, we still need to quantify the effect of such poles because they are approaching the boundary of our set $\{s:Re(s)\geq 0\}$. A similar argument can be made for any set of the form $\{s:Re(s)=\sigma\}=\{\sigma+j\omega:\omega\in[-\infty,+\infty]\}$, which we mentioned in \ref{ssec:goal}.
    
Although proving (or disproving) uniform convergence on the right-half plane for general networks is quite challenging, it is much more straightforward for networks consist of only non-strictly proper nodes, as shown in the following theorem:  
\begin{thm}[Sufficient condition for uniform convergence on right-half plane]\label{thm_unifm_conv_right_plane}
    Let $T(s),\bar{g}(s)$ be defined as in \eqref{eq_T_explict} and \eqref{eq_g_bar}, respectively. Suppose $g_i(s), i\in [n]$ are not strictly proper, $\bar{g}(s)$ is stable, and $\mathcal{N}(s)\leq 1, \forall Re(s)\geq 0$, then we have 
    \ben
        \lim_{\lambda_2(L)\ra +\infty}\sup_{Re(s)\geq 0}\lV T(s)-\frac{1}{n}\bar{g}(s)\one\one^T\rV=0\,.
    \een
\end{thm}
\begin{proof}
    Given $R>0$, we define the following sets:
    \begin{align*}
        &S_1:=\{s\in \compl: Re(s)\geq 0, |s|\leq R\}\,,\\
        &S_2:=\{s\in \compl: Re(s)\geq 0, |s|> R\}\,.
    \end{align*}
    Apparently, $S_1\bigcup S_2=\{s\in \compl: Re(s)\geq 0\}$. Then we can show uniform convergence on right-half plane by proving uniform
    convergence on $S_1,S_2$ respectively:
    
    Firstly, because all $g_i(s)$ are not strictly proper, each $g_i(s)$ converges to some non-zero value as $|s|\ra +\infty$. Then we can choose $R$ large enough so that $\max_{1\leq i\leq n}|g_i^{-1}(s)|\leq M_2, \forall s\in S_2$ for some $M_2>0$. Moreover, $|\bar{g}(s)|<M_1, \forall s\in S_2$ for some $M_1>0$ because $\bar{g}(s)$ is stable. Then the conditions in Lemma \ref{lem_reg_norm_bd} are satisfied, we have
    \be
        \sup_{s\in S_2}\lV T(s)-\frac{1}{n}\bar{g}(s)\one\one^T\rV\leq\frac{\lp M_1M_2+1\rp^2}{F_l\lambda_2(L)-M_2-M_1M_2^2}\,,\label{eq_T_norm_bd_sup_s2}
    \ee
    where $F_l=\inf_{s\in S}|f(s)|$. Taking $\lambda_2(L)\ra +\infty$ on both side of \eqref{eq_T_norm_bd_sup_s2}, we have the uniform convergence on $S_2$.
    
    Secondly, notice that $S_1$ is compact, contains no pole of $\bar{g}(s)$ and has $\mathcal{N}(s)\leq 1,\forall s\in S_1$, the uniform convergence is shown by Theorem \ref{thm_unifm_conv_compact}.
\end{proof}
{\ifthenelse{\boolean{color}}{\color{blue}}{}
\subsection{Connection to Time Domain Response}\label{ssec:to_time}
As discussed in Section \ref{ssec:goal}, we are interested in the uniform convergence of network transfer matrix $T(s)$ because uniform convergence result on the line $\{\sigma+j\omega: \omega\in[-\infty,+\infty]\}$ would allow us to show coherence in time-domain response through the inverse Laplace transform~\cite[Theorem 3.20]{Dullerud2013}
$$
    f(t)=\mathcal{L}^{-1}\{F(s)\}(t)=\frac{1}{2\pi j}\lim_{\omega\ra\infty}\int^{\sigma+j\omega}_{\sigma-j\omega}e^{st}F(s)ds\,.
$$
However, we have seen in Section \ref{ssec:rhp} establishing such uniform convergence is challenging when $g_i(s)$ are strictly proper. Nonetheless, if we assume the input signal decays sufficiently fast in high frequency range, we can prove the following:
\begin{thm}\label{thm_to_time}
    Given $\epsilon>0$ and a real input signal vector with its Laplace transform $U(s)$. Suppose for some $\gamma>0$, $\sigma>0$, $\omega_0>0$, we have 
    \begin{enumerate}
        \item $\sup_{Re(s)>\sigma}\|U(s)\|$ is finite;
        \item $$
         \lim_{\omega\ra \infty}\int_{\sigma+j\omega_0}^{\sigma+j\omega} \|U(s)\|ds\leq \frac{2\pi\epsilon}{6e^\sigma\gamma}\,,
        $$
        \item $\|\bar{g}(s)\|_{\mathcal{H}_\infty}\leq \gamma$;
        \item
        $\|T(s)\|_{\mathcal{H}_\infty}\leq \gamma$, for any Laplacian matrix $L$;
    \end{enumerate}
    Let $y_i(t)$ be the response of $i$-th node when the network input is $U(s)$, and let $\bar{y}(t)$ be the response of $\bar{g}(s)$ to $\frac{\one^T}{n}U(s)$. Then there exists $\lambda=\mathcal{O}(\frac{e^\sigma\gamma}{\epsilon})$, such that If $\lambda_2(L)\geq \lambda$, we have
    $$
        \sup_{t>0}|y_i(t)-\bar{y}(t)|\leq \epsilon\,.
    $$
\end{thm}
\ifthenelse{\boolean{archive}}{The proof is shown in Appendix \ref{app_proof_to_time}}{The proof uses the inverse formula to compute $y_i(t)-\bar{g}(s)$ and breaks the integral on $\{\sigma+j\omega: \omega\in[-\infty,+\infty]\}$ into the integral on low-frequency range $\{\sigma+j\omega: \omega\in[-w_0,+w_0]\}$ and one on high-frequency range. The former can be made arbitrarily small by Theorem \ref{thm_unifm_conv_compact}, and the latter is small from our assumption. We refer interested readers to~\cite{min2021a} for the proof}. The non-asymptotic rate $\mathcal{O}(\frac{e^\sigma\gamma}{\epsilon})$ has hidden constant that implicitly depends on $g_i(s)$,$f(s)$,$U(s)$ and the choice of $w_0$. It is yet independent of $L$. 

Theorem \ref{thm_to_time} made several assumptions: The first one makes sure we can integrate on $\{\sigma+j\omega: \omega\in[-\infty,+\infty]\}$. The second condition requires the input signal decays sufficiently fast in high-frequency range. The third assumption relies on the stability of $\bar{g}(s)$ and can be verified easily. The last assumption is generally hard to verify, but it holds when the network satisfies additional properties. To present the result, we first define the following
\begin{dfn}
    A rational transfer function $g(s)$ is positive real (PR) if
    $$
        Re(g(s))\geq 0, \forall Re(s)>0\,.
    $$
    A rational transfer function $g(s)$ is output strictly passive (OSP) if
    $$
        Re(g(s))\geq \epsilon |g(s)|^2, \forall Re(s)>0\,,
    $$
    for some $\epsilon>0$.
\end{dfn}
With these definition, we have
\begin{thm}\label{thm_passive_to_stable}
    Suppose all $g_i(s),i=1,\cdots,n$ are OSP, and $f(s)$ is PR. There exists $\gamma>0$, such that given any positive semidefinite matrix $L$, we have
    $$
        \|T(s)\|_{\mathcal{H}_\infty}\leq \gamma\,.
    $$
\end{thm}
\ifthenelse{\boolean{archive}}{The proof is shown in Appendix \ref{app_proof_to_time}}{ We refer interested readers to~\cite{min2021a} for the proof}. This theorem shows the stability of the network when $g_i(s)$ are OSP and $f(s)$ is PR, regardless of the network connectivity. The coherence in time-domain response for such network can be understood by Theorem \ref{thm_to_time}.
}

\section{Coherence and Dynamics Concentration in Large-scale Networks}\label{sec:dymC}

Until now we looked into convergence results of $T(s)$ for networks with fixed size $n$. However, one could easily see that such coherence does not depend on the network size $n$. In particular, the right-hand side of \eqref{eq_T_norm_bd} only depends on $n$ via $\lambda_2(L)$ as long as the bounds regarding $g_i(s)$, i.e. $M_1$ and $M_2$ do not scale with respect to $n$. This implies that coherence can emerge as the network size increases. This is the topic of this section.
    
More interestingly, in a stochastic setting where all $g_i(s)$ are unknown transfer functions independently drawn from some distribution, their harmonic mean eventually converges in probability to a deterministic transfer function as the network size increases. Consequently, a large-scale stochastic network concentrates to deterministic a system. We term this phenomenon \emph{dynamics concentration}.
\subsection{Coherence in Large-scale Networks}
To start with, we revise the problem settings to account for variable network size: Let $\{g_i(s), i\in \mathbb{N}_+\}$ be a sequence of transfer functions, and $\{L_n, n\in\mathbb{N}_+\}$ be a sequence of real symmetric Laplacian matrices such that $L_n$ has order $n$, particularly, let $L_1=0$. Then we define a sequence of transfer matrix $T_n(s)$ as
\be
    T_n(s) = \lp I_n+G_n(s)L_n\rp^{-1} G_n(s)\,,\label{eq_T_n_explicit}
\ee
where $G_n(s)= \dg\{g_1(s),\cdots,g_n(s)\}$. This is exactly the same transfer matrix  shown in Fig.\ref{blk_p_n} for a  network of size $n$. We can then define the coherent dynamics for every $T_n(s)$ as
\be
    \bar{g}_n(s) = \lp\frac{1}{n}\sum_{i=1}^ng_i^{-1}(s)\rp^{-1}\,.\label{eq_g_n_bar}
\ee
    
For certain family $\{L_n,n\in\mathbb{N}_+\}$ of large-scale networks, the network algebraic connectivity $\lambda_2(L_n)$ increases as $n$ grows. For example, when $L_n$ is the Laplacian of a complete graph of size $n$ with all edge weights being $1$, we have $\lambda_2(L_n)=n$. As a result, network coherence naturally emerges as the network size grows. Recall that to prove the convergence of $T_n(s)$ to $\frac{1}{n}\bar{g}_n(s)\one\one^T$ for fixed $n$, we essentially seek for $M_1,M_2>0$, such that $|\bar{g}_n(s)|\leq M_1$ and $\max_{1\leq i\leq n}|g_i^{-1}(s)|\leq M_2$ for $s$ in a certain set. If it is possible to find a universal $M_1,M_2>0$ for all $n$, then the convergence results should be extended to arbitrarily large networks, provided that network connectivity increases as $n$ grows. To state this condition formally, we need the notion of uniform boundedness for a family of functions.
\begin{dfn}
    Let $\{g_i(s), i\in I\}$ be a family of complex functions indexed by $I$. Given $S\subset \compl$, $\{g_i(s), i\in I\}$ is uniformly bounded on $S$ if
    \ben
        \exists M>0\quad s.t.\quad |g_i(s)|\leq M,\quad \forall i\in I,\ \forall s\in S\,. 
    \een
\end{dfn}
Now we are ready to show uniform convergence of $T_n(s)$:  
\begin{thm}\label{thm_unifm_conv_reg_dymC}
    Let $T_n(s)$ and $\bar{g}_n(s)$ be defined as in \eqref{eq_T_n_explicit} and \eqref{eq_g_n_bar}, respectively. Suppose $\lambda_2(L_n)\ra +\infty$ as $n\ra \infty$. If both $\{g_i^{-1}(s),i\in\mathbb{N}_+\}$ and $\{\bar{g}_n(s),n\in\mathbb{N}_+\}$ are uniformly bounded on a set $S\subset \compl$. then we have
    \ben
        \lim_{n\ra \infty}\sup_{s\in S}\lV T_n(s)-\frac{1}{n}\bar{g}_n(s)\one\one^T\rV=0\,.
    \een
\end{thm}
\begin{proof}
    Since both $\{g_i^{-1}(s),i\in\mathbb{N}_+\}$ and $\{\bar{g}_n(s),n\in\mathbb{N}_+\}$ are uniformly bounded on $S$, $\exists M_1,M_2>0$ s.t. $|\bar{g}_n(s)|\leq M_1$ and $\max_{1\leq i\leq n}|g_i^{-1}(s)|\leq M_2$ for every $n\in \mathbb{N}_+$ and $s\in S$. By Lemma \ref{lem_reg_norm_bd}, $\forall n\in\mathbb{N}_+$,
    \be
        \sup_{s\in S}\lV T_n(s)-\frac{1}{n}\bar{g}_n(s)\one\one^T\rV\leq 
        \frac{\lp M_1M_2+1\rp^2}{F_l\lambda_2(L_n)-M_2-M_1M_2^2}\,,\label{eq_T_norm_bd_dymC}
    \ee
    where $F_l=\inf_{s\in S}|f(s)|$. We already assumed that $\lambda_2(L_n)\ra +\infty$ as $n\ra +\infty$, therefore the proof is finished by taking $n\ra +\infty$ on both sides of \eqref{eq_T_norm_bd_dymC}.
\end{proof}
\begin{rem}
    Similarly to Theorem \ref{thm_unifm_conv_reg_compact}, uniform convergence is achieved on a set away from  zeros or poles of $\bar{g}(s)$. The uniform boundedness condition is preventing any point in the closure of $S$ from asymptotically becoming a zero of any $g_i^{-1}(s)$ or a pole of $\bar{g}(s)$ as $n$ increases.
\end{rem}
\subsection{Dynamics Concentration in Large-scale Networks}
Now we consider the cases where the node dynamics are unknown (stochastic). For simplicity, we constraint our analysis to the setting where the node dynamics are independently sampled from the same random rational transfer function with all or part of the coefficients are random variables, i.e. the nodal transfer functions are of the form
\begin{equation}\label{eq_random_tf}
    g_i(s) \sim \frac{b_ms^m+\dots b_1s+b_0}{a_ls^l+\dots a_1s+a_0}\,,
\end{equation}
for some $m,l>0$, where $b_0,\cdots,b_m$, $a_0,\cdots, a_l$ are random variables. 

To formalize the setting, we firstly define the random transfer function to be sampled. Let $\Omega=\mathbb{R}^d$ be the sample space, $\mathcal{F}$ the Borel $\sigma$-field of $\Omega$, and $\prob$ a probability measure on $\Omega$. A sample $w\in\Omega$ thus represents a $d$-dimensional vector of coefficients. We then define a random rational transfer function $g(s,w)$ on $(\Omega,\mathcal{F},\mathbb{P})$ such that all or part of the coefficients of $g(s,w)$ are random variables. Then for any $w_0\in \Omega$, $g(s,w_0)$ is a rational transfer function.
    
Now consider the probability space $(\Omega^\infty,\mathcal{F}^\infty,\mathbb{P}^\infty)$. Every $\mathbf{w}\in \Omega^\infty$ give an instance of samples drawn from our random transfer function: $$g_i(s,w_i):=g(s,w_i),i\in\mathbb{N}_+\,,$$ where $w_i$ is the $i$-th element of $\mathbf{w}$. By construction, $g_i(s,w_i), i\in\mathbb{N}_+$ are i.i.d. random transfer functions. Moreover, for every $s_0\in\compl$, $g_i(s_0,w_i),i\in\mathbb{N}_+$ are i.i.d. random complex variables taking values in the extended complex plane (presumably taking value $\infty$).

Now given $\{L_n, n\in\mathbb{N}_+\}$ a sequence of $n\by n$ real symmetric Laplacian matrices, consider the random network of size $n$ whose nodes are associated with the dynamics $g_i(s,w_i),i=1,2,\cdots,n$ and coupled through $L_n$. The transfer matrix of such a network is given by
\be\label{eq_T_stoch}
    T_n(s,\mathbf{w})=(I_n+G_n(s,\mathbf{w})L_n)^{-1}G_n(s,\mathbf{w})\,,
\ee
where $G_n(s,\mathbf{w})=\dg\{g_1(s,w_1),\cdots,g_n(s,w_n)\}$. 
    
Then under this setting, the coherent dynamics of the network is given by
\be\label{eq_g_bar_stoch}
    \bar{g}(s,\mathbf{w}) = \lp \frac{1}{n}\sum_{i=1}^ng_i^{-1}(s,w_i)\rp^{-1}\,.
\ee
    
Now given a compact set $S\subset \compl$ of interest, and assuming suitable conditions on the distribution of $g(s,w)$, we expect that the random coherent dynamics $\bar{g}(s,\mathbf{w})$ would converge uniformly in probability to its expectation
\be\label{eq_expc_g}
    \hat{g}(s)=\lp\expc g^{-1}(s,w))\rp^{-1}:=\lp \int_{\Omega}g^{-1}(s,w)d\prob(w)\rp^{-1}\,,       
\ee
for all $s\in S$, as $n\ra \infty$. The following Lemma provides a sufficient condition for this to hold. 
\begin{lem}\label{lem_unifm_prob_conv_compact}
    Consider the probability space $(\Omega^\infty,\mathcal{F}^\infty,\prob^\infty)$. Let $\bar{g}_n(s,\mathbf{w})$ and $\hat{g}(s)$ be defined as in \eqref{eq_g_bar_stoch} and \eqref{eq_expc_g}, respectively, and given a compact set $S\subset \compl$, let the following conditions hold:
    \begin{enumerate}
        \item $g^{-1}(s,w)$ is uniformly bounded on $S\times \Omega$;
        \item $\{\bar{g}_n(s,\mathbf{w}),n\in\mathbb{N}_+\}$ are uniformly bounded on $S\times \Omega^\infty$;
        \item $\exists L>0$ s.t. $|g_1^{-1}(s_1,w)-g_1^{-1}(s_2,w)|\leq L|s_1-s_2|$, $\forall w\in\Omega,\forall s_1,s_2\in S$;
        \item $\hat{g}(s)$ is uniformly continuous.
    \end{enumerate}
    Then, $\forall \epsilon>0$, we have
    \ben
        \lim_{n\ra \infty}\prob\lp\sup_{s\in S}\lV \frac{1}{n}\bar{g}_n(s,\mathbf{w})\one\one^T-\frac{1}{n}\hat{g}(s)\one\one^T\rV\geq \epsilon\rp=0\,.
    \een
\end{lem}
\ifthenelse{\boolean{archive}}{The proof is shown in  Appendix \ref{app_proof_lem_unifm_prob_conv_compact}.}{The proof follows the standard procedure~\cite{Newey1991} for showing the uniform stochastic convergence of a random function, for which we refer interested readers to~\cite{min2021a}.} This lemma suggests that our coherent dynamics $\bar{g}_n(s,\mathbf{w})$, as $n$ increases, converges uniformly on $S$ to its expected version $\hat{g}(s)$. Then provided that the coherence is obtained as the network size grows, we would expect that the random transfer matrix $T_n(s, \mathbf{w})$ to concentrate to a deterministic one $\frac{1}{n}\hat{g}(s)\one\one^T$, as the following theorem shows.
\begin{thm}\label{thm_unifm_prob_conv_compact}
    Given probability space $(\Omega^\infty,\mathcal{F}^\infty,\prob^\infty)$. Let $T_n(s,\mathbf{w})$ and $\hat{g}(s)$ be defined as in \eqref{eq_T_stoch} and \eqref{eq_expc_g}, respectively. Suppose $\lambda_2(L_n)\ra +\infty$ as $n\ra +\infty$. Given a compact set $S\subset \compl$, if all the conditions in Lemma \ref{lem_unifm_prob_conv_compact} hold, then $\forall \epsilon>0$, we have
    \ben
        \lim_{n\ra \infty}\prob\lp\sup_{s\in S}\lV T_n(s,\mathbf{w})-\frac{1}{n}\hat{g}(s)\one\one^T\rV\geq \epsilon\rp=0\,.
    \een
\end{thm}
\ifthenelse{\boolean{archive}}{
\begin{proof}
    Firstly, notice that
    \begin{align*}
        &\;\prob\lp\sup_{s\in S}\lV T_n(s,\mathbf{w})-\frac{1}{n}\hat{g}(s)\one\one^T\rV\geq \epsilon\rp\\
        \leq &\; \prob\lp\sup_{s\in S}\lV T_n(s,\mathbf{w})-\frac{1}{n}\bar{g}_n(s)\one\one^T\rV+\right.\\
        &\; \quad\quad\qquad\left.\sup_{s\in S}\lV \frac{1}{n}\bar{g}_n(s,\mathbf{w})\one\one^T-\frac{1}{n}\hat{g}(s)\one\one^T\rV\geq \epsilon\rp\\
        \leq &\; \prob\lp\sup_{s\in S}\lV T_n(s,\mathbf{w})-\frac{1}{n}\bar{g}_n(s,\mathbf{w})\one\one^T\rV\geq \frac{\epsilon}{2}\rp+\\
        &\;\quad\quad\qquad\prob\lp\sup_{s\in S}\lV\frac{1}{n}\bar{g}_n(s,\mathbf{w})\one\one^T-\frac{1}{n}\hat{g}(s)\one\one^T\rV\geq \frac{\epsilon}{2}\rp\,.
    \end{align*}
    The second term converges to $0$ as $n\ra +\infty$ by Lemma \ref{lem_unifm_prob_conv_compact}. For the first term, we show below that it becomes exactly $0$ for large enough $n$. Still, we assume $\{\bar{g}_n(s,\mathbf{w})\}$ and $\{g_i^{-1}(s,\mathbf{w})\}$ are uniformly bounded on $S$ by $M_1,M_2>0$ respectively. By Lemma \ref{lem_reg_norm_bd}, choosing large enough $n$ s.t.
    \begin{align*}
        &\;\prob\lp\sup_{s\in S}\lV T_n(s,\mathbf{w})-\frac{1}{n}\bar{g}_n(s,\mathbf{w})\one\one^T\rV\geq \frac{\epsilon}{2}\rp\\
        \leq &\; \prob\lp \frac{\lp M_1M_2+1\rp^2}{F_l\lambda_2(L_n)-M_2-M_1M_2^2}\geq \frac{\epsilon}{2}\rp\,,
    \end{align*}
    then we can choose even larger $n$ such that the probability on the right-hand side is $0$ because $\lambda_2(L_n)\ra +\infty$ as $n\ra \infty$.
\end{proof}}{We refer the readers to~\cite{min2021a} for the proof.}

\begin{rem}
    Lemma \ref{lem_unifm_prob_conv_compact} requires $g^{-1}(s,w)$ to be uniformly bounded on $S\times \Omega$. That is, for $s_0\in S$, $g^{-1}(s_0,w)$ is a bounded complex random variable. In \cite{min2019cdc}, a weaker condition, that $g^{-1}(s_0,w)$ is a sub-Gaussian complex random variable, is considered. This allows to show that point-wise convergence in probability can be achieved whenever $\lambda_2(L_n)$ grows polynomially in $n$. 
\end{rem}
    
In summary, because the coherent dynamics of the tightly-connected network is given by the harmonic mean of all node dynamics $g_i(s)$, it concentrates to its harmonic expectation $\hat g(s)$ as the network size grows. As a result, in practice, the coherent behavior of large-scale tightly-connected networks depends on the empirical distribution of $g_i(s)$, i.e. a collective effect of all node dynamics rather than every individual node dynamics. For example, two different realizations of large-scale network with dynamics $T_n(s,\mathbf{w})$ exhibit similar coherent behavior with high probability, in spite of the possible substantial differences in individual node dynamics.

\section{Application: Aggregate Dynamics of Synchronous Generator Networks}\label{sec:exmples}
In this section, 
we apply our analysis to investigate coherence in power networks. For coherent generator groups, we find that $\frac{1}{n}\hat{g}(s)$ generalized typical aggregate generator models which are often used for model reduction in power networks~\cite{Chow2013}. Moreover, we show that heterogeneity in generator dynamics usually leads to high-order aggregate dynamics, making it challenging to find a reasonably low-order approximation. 
Consider the transfer matrix of power generator networks~\cite{Paganini2019tac} linearized around its steady-state point, given by the following block diagram:
\begin{figure}[ht]
    \centering
    \includegraphics[height=2.5cm]{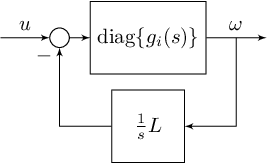}
    \caption{Block Diagram of Linearized Power Networks}\label{blk_power}
\end{figure}
This is exactly the block structure shown in Fig. \ref{blk_p_n} with $f(s)=\frac{1}{s}$. Here, the network output, i.e., the frequency deviation of each generator, is denoted by $\omega$. Generally, the $g_i(s)$ are modeled as strictly positive real transfer functions and we assume $L$ is connected. 
    
\subsection{Coherence Analysis}
We utilize the convergence results from previous sections to characterize the coherent behavior of such networks. We still denote the coherent dynamics as $\bar{g}(s)$ defined in \eqref{eq_g_bar}.
    
Firstly, Notice that $s=0$ is a pole of $f(s)=\frac{1}{s}$, then by Theorem  \ref{thm_ptw_singular_f_pole}, $T(0)$ is exactly $\frac{1}{n}\bar{g}(0)\one\one^T$, which suggests that in steady state, frequency outputs are the same for all generators. Moreover, another consequence of Theorem \ref{thm_ptw_singular_f_pole} is as follows.
\begin{clm}
    Given fixed $L$, $\forall\epsilon>0$, $\exists \eta_c>0$ such that
    \ben
        \sup_{\eta\in[-\eta_c,\eta_c]}\lV T(j\eta)-\frac{1}{n}\bar{g}(j\eta)\one\one^T\rV<\epsilon\,.
    \een
\end{clm}
\begin{proof}
    This is the direct application of Theorem \ref{thm_ptw_singular_f_pole} according to the definition of the limit.
\end{proof}
The claim suggests that the network is naturally coherent in the low-frequency range. In other words, for any fixed network $L$, there is a low-frequency range such that the network responds coherently to disturbances in that frequency range. Additionally, we know that such frequency range can be arbitrarily wide, given sufficiently large network connectivity, suggested by the following claim.
\begin{clm}
    $\forall \eta_{c}>0$, we have
    \ben
        \lim_{\lambda_2(L)\ra \infty}\sup_{\eta\in[-\eta_c,\eta_c]}\lV T(j\eta)-\frac{1}{n}\bar{g}(j\eta)\one\one^T\rV=0\,.
    \een
\end{clm}
\begin{proof}
    Provided that $g_i(s)$ are strictly positive real~\cite{khalil2002nonlinear}, this is a direct application of Theorem \ref{thm_unifm_conv_reg_compact}.
\end{proof}
Loosely speaking, the generator network is coherent for certain low-frequency range and the width of such frequency range increases when the network is more connected ($\lambda_2(L)$ increases).
    
Furthermore, notice that $\infty$ can be regarded as a zero of $f(s)$, around which the network effect diminishes, i.e. for sufficiently large $\eta$, the effective algebraic connectivity $|f(j\eta)\lambda_2(L)|=\frac{\lambda_2(L)}{|\eta|}$ can be arbitrarily small. Hence given any fixed $L$, there is a high-frequency range such that the network does not exhibit coherence under disturbances within such range. 

\begin{figure}[ht]
    \centering
    \includegraphics[width=3.49in]{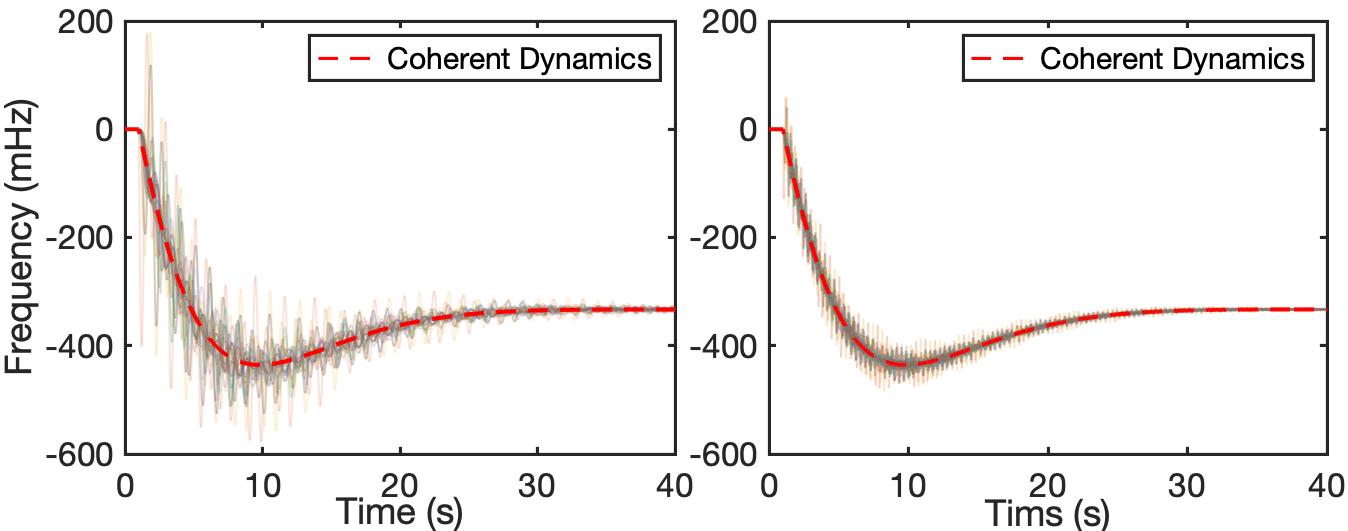}
    \caption{Step responses of Icelandic Grid without (Left) and with (Right) connectivity $\lambda_2(L)$ scaled up. The responses of coherent dynamics $\bar{g}(s)$ are shown in red dashed lines.\label{fig_coherence}}
\end{figure}
    
\begin{figure}[ht]
    \centering
    \includegraphics[width=3.49in]{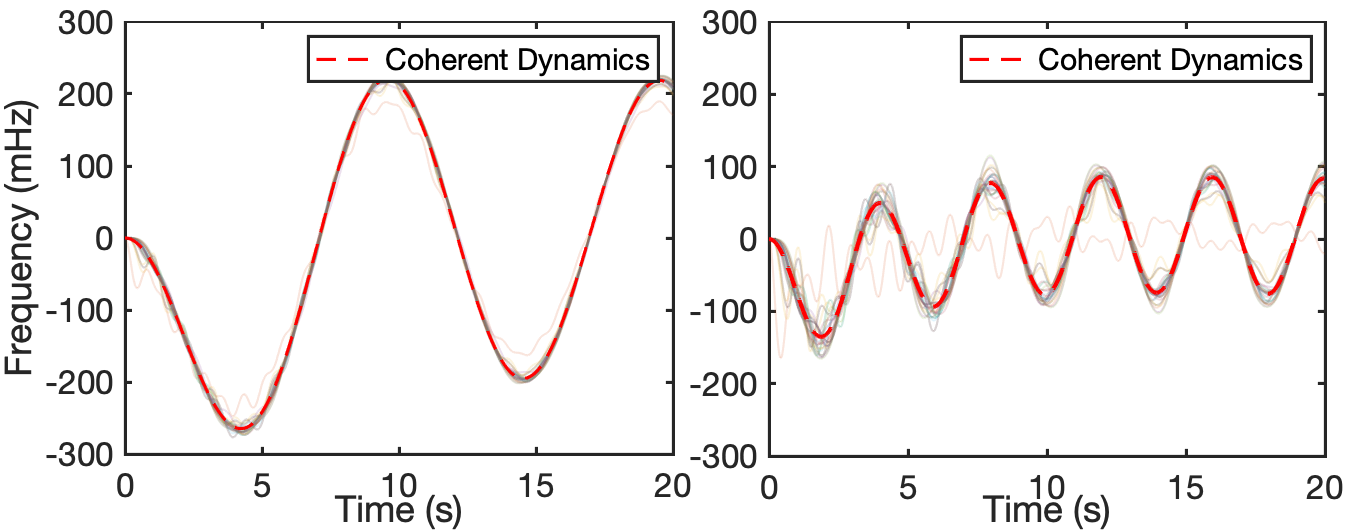}
    \caption{Responses of Icelandic Grid under sinusoidal disturbances of frequency $w_\mathrm{low}=0.1\mathrm{ rad/s}$ (Left) and $w_\mathrm{high}=0.25\mathrm{ rad/s}$ (Right). The responses of coherent dynamics $\bar{g}(s)$ are shown in red dashed lines.\label{fig_freq_test}}
\end{figure}
We verify our analysis with simulations on the Icelandic power grid~\cite{iceland}. As shown in Fig.~\ref{fig_coherence}, the network step response is more coherent, i.e. response of every single node (generator) is getting closer to the one of the coherent dynamics $\bar{g}(s)$, when the network connectivity is scaled up. For the plot on the right, the connectivity $\lambda_2(L)$ is scaled up to 10 times the original. Then Fig.~\ref{fig_freq_test} shows such coherence is also frequency-dependent, where the generators respond less coherently under disturbances of  higher frequency. 
    
The discussions and simulations above suggest that the coherent dynamics $\bar g(s)$ characterize well the overall response/behavior of generators. This leads to a general methodology to analyze the aggregate dynamics of such networks, that we describe next.  
\subsection{Aggregate Dynamics of Generator Networks}
    
Let $$g_\mathrm{aggr}(s):=\frac{1}{n}\bar{g}(s)=\lp \sum_{i=1}^ng_i^{-1}(s)\rp\,.$$ Our analysis suggests that the transfer function $T(s)$ representing a network of  generators is close $g_\mathrm{aggr}(s)\one\one^T$ within the low-frequency range, for sufficiently high network connectivity $\lambda_2(L)$. We can also view $g_\mathrm{aggr}(s)$ as the aggregate generator dynamics, in the sense that it takes the sum of disturbances $\one^Tu=\sum_{i=1}^nu_i$ as its input, and its output represents the coherent response of all generators.
    
Such a notion of aggregate dynamics is important in modeling large-scale power networks\cite{Chow2013}. Generally speaking, one seeks to find an aggregate dynamic model for a group of generators using the same structure (transfer function) as individual generator dynamics, i.e. when generator dynamics are modeled as $g_i(s)=g(s;\theta_i)$, where $\theta_i$ is a vector of parameters representing physical properties of each generator, existing works~\cite{Germond1978,Guggilam2018,Apostolopoulou2016} propose methods to find aggregate dynamics of the form $g(s;\theta_\mathrm{aggr})$ for certain structures of $g(s;\theta)$. Our $g_\mathrm{aggr}(s)$ justifies their choices of $\theta_\mathrm{aggr}$, as shown in the following example.
    
\begin{example}
    For generators given by the swing model $g_i(s) =\frac{1}{m_is+d_i}\,,$
    where $m_i,d_i$ are the inertia and damping of generator $i$, respectively. The aggregate dynamics are
    \be
        g_\mathrm{aggr}(s) = \frac{1}{m_\mathrm{aggr}s+d_\mathrm{aggr}}\,,\label{eq_aggr_dym_sw}
    \ee
    where $m_\mathrm{aggr} = \sum_{i=1}^nm_i$ and $d_\mathrm{aggr} = \sum_{i=1}^nd_i$.
\end{example}
Here the parameters are $\theta=\{m,d\}$. The aggregate model given by \eqref{eq_aggr_dym_sw} is consistent with the existing approach of choosing inertia $m$ and damping $d$ as the respective sums over all the coherent generators. 
    
However, as we show in the next example, when one considers more involved models,  it is challenging to find parameters $\theta_\mathrm{aggr}$ that accurately fit the aggregate dynamics.
\begin{example}
    For generators given by the swing model with turbine droop
    \be \label{eq_single_generator}
        g_i(s) = \frac{1}{m_is+d_i+\frac{r_i^{-1}}{\tau_is+1}}\,,
    \ee
    where $r_i^{-1}$ and  $\tau_i$ are the droop coefficient and turbine time constant of generator $i$, respectively. The aggregate dynamics are given by 
    \be
        g_\mathrm{aggr}(s) = \frac{1}{m_\mathrm{aggr}s+d_\mathrm{aggr}+\sum_{i=1}^n\frac{r_i^{-1}}{\tau_is+1}}\,.\label{eq_aggr_dym_sw_tb}
    \ee
\end{example}
Here the parameters are $\theta=\{m,d,r^{-1},\tau\}$. This example illustrates, in particular, the difficulty in aggregating generators with heterogeneous turbine time constants. When all generators have the same turbine time constant $\tau_i=\tau$, then $g_\mathrm{aggr}(s)$ in \eqref{eq_aggr_dym_sw_tb} reduces to the typical effective machine model
\ben
    g_\mathrm{aggr}(s) = \frac{1}{m_\mathrm{aggr}s+d_\mathrm{aggr}+\frac{r_\mathrm{aggr}^{-1}}{\tau s+1}}\,,
\een
where $r_\mathrm{aggr}^{-1}=\sum_{i=1}^nr_i^{-1}$, i.e. the aggregation model is still obtained by choosing parameters $\{m,d, r^{-1}\}$ as the respective sums of their individual values. 
    
If the $\tau_i$  are heterogeneous, then $g_\mathrm{aggr}(s)$ is a high-order transfer function and cannot be accurately represented by a single generator model. The aggregation of generators essentially asks for a low-order approximation of $g_\mathrm{aggr}(s)$. Our analysis reveals the fundamental limitation of using conventional approaches seeking aggregate dynamics with the same structure of individual generators. Furthermore, by characterizing the aggregate dynamics in the explicit form $g_\mathrm{aggr}(s)$, one can develop more accurate low-order approximation~\cite{min2020lcss}. Lastly, we emphasize that our analysis does not depend on a specific model of generator dynamics $g_i(s)$, hence it provides a general methodology to aggregate coherent generator networks.

\section{Conclusions}\label{sec:conclusion}    
This paper provides various convergence results of the transfer matrix $T(s)$ for a tightly-connected network. The analysis leads to useful characterizations of coordinated behavior and justifies the relation between network coherence and network effective algebraic connectivity. Our results suggest that network coherence is also a frequency-dependent phenomenon, which is numerically illustrated in generator networks. Lastly, concentration results for large-scale tightly-connected networks are presented, revealing the exclusive role of the statistical distribution of node dynamics in determining the coherent dynamics of such networks.  

As we already see from the numerical example of synchronous generator networks, when the network is coherent, the entire network can be represented by its aggregate dynamics, because the network has all non-zero eigenvalue being sufficiently large to make the dynamics associated with these eigenvalues negligible. One could extend such observation to less coherent networks which has large eigenvalues except for a few small eigenvalues. We can utilize the same frequency domain analysis to these networks, but instead of only keeping the coherent dynamics, we retain the dynamic associated with all the small eigenvalues (These dynamics are generally coupled due to the heterogeneity in node dynamics). In this way, we can accurately capture the important modes in the network dynamics while significantly reduce the complexity of our model. 

Furthermore, for large-scale networks with multiple coherent groups, or communities in a more general sense, one could model the inter-community interactions by replacing the dynamics of each community with its coherent one, or more generally, a reduced one. Although clustering, i.e. finding communities, for homogeneous networks can be efficiently done by various methods such as spectral clustering~\cite{bach2004learning}~\cite{Belkin2001}, it is still open for research to find multiple coherent groups in heterogeneous dynamical networks.

\appendix
\setcounter{equation}{0}
\def\theequation{\thesubsection.\arabic{equation}}
\ifthenelse{\boolean{archive}}{
\subsection{Proof of Lemma \ref{lem_grd_Lap_eig_bd}}
\label{app_proof_lem_grd_Lap_eig_bd}
\setcounter{equation}{0}
\def\theequation{\thesubsection.\arabic{equation}}
\begin{proof}
To avoid confusion, here we let $\one_n$ denote all one vector $\one$ of length $n$.

Without loss of generality, assume $\mathcal{I}=\{1,\cdots,m\}$. Then $\Tilde{L}$ is the principal submatrix of $L$ by removing first $m$ rows and columns.

By \cite[3.5]{fiedler73}, the matrix
\ben
    L-\lambda_2(L)\lp I-\frac{1}{n}\one_n\one_n^T\rp\,,
\een
is positive semi-definite. Now let $v=\bmt 
0\\ \Tilde{v}_1 
\emt $ where $\Tilde{v}_1$ is the unit eigenvector correspoding to $\lambda_1(\Tilde{L})$, we have
\begin{align*}
    &\;v^T\lp L-\lambda_2(L)\lp I-\frac{1}{n}\one_n\one_n^T\rp\rp v\\
    =&\;\tilde{v}_1^T\tilde{L}\tilde{v}_1-\lambda_2(L)\lp 1-\frac{1}{n}(\one_{n-m}^T\Tilde{v}_1)^2\rp\\
    =&\;\lambda_1(\Tilde{L})-\lambda_2(L)\lp 1-\frac{1}{n}(\one_{n-m}^T\Tilde{v}_1)^2\rp\geq0\,.
\end{align*}
We also have $\one_{n-m}^T\Tilde{v}_1\leq \|\one_{n-m}\|\|\Tilde{v}_1\|=\sqrt{n-m}$, then the desired lower bound on $\lambda_1(\Tilde{L})$ is obtained:
\ben
    \lambda_1(\Tilde{L})\geq \lambda_2(L)\lp 1-\frac{1}{n}(\one_{n-m}^T\Tilde{v}_1)^2\rp\geq \frac{m}{n}\lambda_2(L)\,.
\een
\end{proof}}{}

{
\subsection{Convergence of Normalized Transfer Function at Poles of $\bar{g}(s)$}\label{app_lim_dir_pole}

We present here the formal statement regarding the Remark \ref{rem_pole} on the convergence of normalized transfer function at poles of $\bar{g}(s)$.
\begin{prop}
    Let $T(s)$ and $\bar{g}(s)$ be defined as in \eqref{eq_T_explict} and \eqref{eq_g_bar}, respectively. If $s_0\in\compl$ is a pole of $\bar{g}(s)$ and a generic point of $f(s)$, and additionally, we assume $\|\tilde{\Lambda}/\lambda_2(L)-\Lambda_{\mathrm{lim}}\|\ra 0$ as $\lambda_2(L)\ra +\infty$, then
    $$
        \lim_{\lambda_2(L)\ra+\infty} \lV \frac{T(s_0)}{\|T(s_0)\|}-\frac{1}{n}\gamma(\Lambda_\mathrm{lim})\one\one^T\rV= 0\,,
    $$
    where $\gamma(\Lambda_\mathrm{lim} )=\frac{h_{12}^T\Lambda^{-1}_\mathrm{lim}h_{12}}{|h_{12}^T\Lambda^{-1}_\mathrm{lim}h_{12}|}$.
\end{prop}
Here $\gamma(\Lambda_\mathrm{lim} )$ is well-defined since $\Lambda_\mathrm{lim}$ must be positive definite, otherwise it contradicts the fact that $\tilde{\Lambda}/\lambda_2(L)-I\succeq 0$. \ifthenelse{\boolean{archive}}{
\begin{proof}
    Since $$\frac{T(s_0)}{\|T(s_0)\|}=\frac{T(s_0)}{a}\cdot\frac{|a|}{\|T(s_0)\|}\cdot\frac{a}{|a|}\,,$$
    where $a$ is defined as in \eqref{eq_H_inv_pole}.
    The desired result follows once we show that $\lV\frac{T(s_0)}{a}-\frac{1}{n}\one\one^T\rV\ra 0$, $\frac{|a|}{\|T(s_0)\|}\ra 1$ and $\frac{a}{|a|}\ra \gamma(\Lambda_\mathrm{lim})$. Recall that $T(s_0)=VH^{-1}V^T$ and 
    $$H^{-1} =
        a\bmt 1\\
        -H_{22}^{-1}h_{21}\emt\bmt1 &-h^T_{21}H_{22}^{-1}\emt+\bmt 0&0\\
        0 & H_{22}^{-1}\emt\,.$$
    For the first limit, we have
    \begin{align*}
        &\;\lV\frac{T(s_0)}{a}-\frac{1}{n}\one\one^T\rV\\
        =&\; \lV \frac{H^{-1}}{a}-e_1e_1^T\rV\\
        =&\; \lV\bmt 0\\
        -H_{22}^{-1}h_{21}\emt\bmt 0 &-h^T_{21}H_{22}^{-1}\emt+\bmt 0&0\\
        0 & H_{22}^{-1}\emt/a\rV\\
        =&\; \|h^T_{21}H^{-1}_{22}\|^2+\|h_{21}^TH_{22}^{-1}h_{21}H^{-1}_{22}\|\\
        \leq&\; 2\|h_{21}\|^2\|H^{-1}_{22}\|^2\leq \frac{2M^2}{(|f(s_0)|\lambda_2(L)-M)^2}\ra 0\,.
    \end{align*}
    The second limit comes from 
    \begin{align*}
       &\; \|T(s_0)\|/|a|\\
       =&\; \lV a\bmt 1\\
        -H_{22}^{-1}h_{21}\emt\bmt 1 &-h^T_{21}H_{22}^{-1}\emt+\bmt 0&0\\
        0 & H_{22}^{-1}\emt\rV/|a|\\
        \leq &\; \lV \bmt 1\\
        -H_{22}^{-1}h_{21}\emt\bmt 1 &-h^T_{21}H_{22}^{-1}\emt\rV +\|H^{-1}_{22}\|/|a|\\
        \leq &\; 1+\|H_{22}^{-1}h_{21}\|^2+\|H_{22}^{-1}\|^2\|h_{21}\|^2\,,
    \end{align*}
    and 
    \begin{align*}
       &\; \|T(s_0)\|/|a|\\
       \geq &\; \lV \bmt 1\\
        -H_{22}^{-1}h_{21}\emt\bmt 1 &-h^T_{21}H_{22}^{-1}\emt\rV -\|H^{-1}_{22}\|/|a|\\
        \geq &\; 1+\|H_{22}^{-1}h_{21}\|^2-\|H_{22}^{-1}\|^2\|h_{21}\|^2\,,
    \end{align*}
    along with the fact that both the upper and the lower bound converge to $1$. To see the limit, notice that 
    \begin{align*}
        \lv\|H_{22}^{-1}h_{21}\|^2+\|H_{22}^{-1}\|^2\|h_{21}\|^2\rv\leq&\; 2\|H_{22}^{-1}\|^2\|h_{21}\|^2\\
        \lv\|H_{22}^{-1}h_{21}\|^2-\|H_{22}^{-1}\|^2\|h_{21}\|^2\rv
        \leq&\; \|H_{22}^{-1}h_{21}\|^2+\|H_{22}^{-1}\|^2\|h_{21}\|^2\\
        \leq&\; 2\|H_{22}^{-1}\|^2\|h_{21}\|^2\,,
    \end{align*}
    and
    \ben
        2\|H_{22}^{-1}\|^2\|h_{21}\|^2\leq \frac{2M^2}{(|f(s_0)|\lambda_2(L)-M)^2}\ra 0\,.
    \een
    The last limit is obtained from
    \begin{align*}
        \frac{a}{|a|}&=\;\frac{h_{12}^TH_{22}^{-1}h_{12}}{|h_{12}^TH_{22}^{-1}h_{12}|}\\
        &=\; \gamma(\Lambda_\mathrm{lim})\frac{h_{12}^TH_{22}^{-1}h_{12}}{f^{-1}(s_0)\lambda_2^{-1}(L)h_{12}^T\Lambda_\mathrm{lim}^{-1}h_{12}}\frac{|f^{-1}(s_0)\lambda_2^{-1}(L)h_{12}^T\Lambda_\mathrm{lim}^{-1}h_{12}|}{|h_{12}^TH_{22}^{-1}h_{12}|}\,,
    \end{align*}
    and the fact that $\frac{h_{12}^TH_{22}^{-1}h_{12}}{f^{-1}(s_0)\lambda_2^{-1}(L)h_{12}^T\Lambda_\mathrm{lim}^{-1}h_{12}}\ra 1$.
    
To show the limit of the ratio in the end, notice that
\begin{align*}
    &\;\lv \frac{h_{12}^TH_{22}^{-1}h_{12}}{f^{-1}(s_0)\lambda_2^{-1}(L)h_{12}^T\Lambda_\mathrm{lim}^{-1}h_{12}}-1\rv\\
    =&\;\lv \frac{h_{12}^T(V_{\perp}^T\tilde{G}(s_0)V_\perp+f(s_0)\tilde{\Lambda})^{-1}h_{12}}{f^{-1}(s_0)\lambda_2^{-1}(L)h_{12}^T\Lambda_\mathrm{lim}^{-1}h_{12}}-1\rv\\
    =&\; \lv \frac{h_{12}^T\lp \frac{V_{\perp}^T\tilde{G}(s_0)V_\perp}{f(s_0)\lambda_2(L)}+\frac{\tilde{\Lambda}}{\lambda_2(L)}\rp^{-1}h_{12}}{h_{12}^T\Lambda_\mathrm{lim}^{-1}h_{12}}-1\rv\\
    =&\;\lv \frac{h_{12}^T\lhp\lp \frac{V_{\perp}^T\tilde{G}(s_0)V_\perp}{f(s_0)\lambda_2(L)}+\frac{\tilde{\Lambda}}{\lambda_2(L)}\rp^{-1}-\Lambda_\mathrm{lim}^{-1}\rhp h_{12}}{h_{12}^T\Lambda_\mathrm{lim}^{-1}h_{12}}\rv\\
    \leq &\;\frac{\|h_{12}\|^2}{|h_12^T\Lambda_\mathrm{lim}h_{12}|}\lV\lp \frac{V_{\perp}^T\tilde{G}(s_0)V_\perp}{f(s_0)\lambda_2(L)}+\frac{\tilde{\Lambda}}{\lambda_2(L)}\rp^{-1}-\Lambda_\mathrm{lim}^{-1}\rV\\
    \leq&\; \frac{\|h_{12}\|^2}{|h_12^T\Lambda_\mathrm{lim}h_{12}|}\frac{\|\Lambda_\mathrm{lim}^{-1}\|^2\|\Delta\|}{1-\|\Lambda_\mathrm{lim}^{-1}\|\|\Delta\|}\,,
\end{align*}
where $$\Delta=\frac{V_{\perp}^T\tilde{G}(s_0)V_\perp}{f(s_0)\lambda_2(L)}+\frac{\tilde{\Lambda}}{\lambda_2(L)}-\Lambda_\mathrm{lim}\,.$$
Here we use the fact that for some invertible square matrices $A,B$, we have
\begin{align*}
    \|(A+B)^{-1}-A^{-1}\|&=\;\|(A+B)^{-1}(A+B-A)A^{-1}\|\\
    &\leq\; \|(A+B)^{-1}\|\|B\|\|A^{-1}\|\\
    (\text{Lemma \ref{lem_bd_norm_mat_inv}})&\leq\;\frac{\|A^{-1}\|}{1-\|A^{-1}\|\|B\|}\|B\|\|A^{-1}\|\\
    &=\;\frac{\|A^{-1}\|^2\|B\|}{1-\|A^{-1}\|\|B\|}\,,
\end{align*}
provided that $1-\|A^{-1}\|\|B\|>0$. To finish the proof, we see that
\begin{align*}
    \|\Delta\|&\leq\;\lV\frac{V_{\perp}^T\tilde{G}(s_0)V_\perp}{f(s_0)\lambda_2(L)}\rV+\lV\frac{\tilde{\Lambda}}{\lambda_2(L)}-\Lambda_\mathrm{lim} \rV\\
    &\leq\; \frac{M}{f(s_0)\lambda_2(L)}+\lV\frac{\tilde{\Lambda}}{\lambda_2(L)}-\Lambda_\mathrm{lim} \rV\ra 0\,.
\end{align*}
This shows that
\begin{align*}
    &\;\lv \frac{h_{12}^TH_{22}^{-1}h_{12}}{f^{-1}(s_0)\lambda_2^{-1}(L)h_{12}^T\Lambda_\mathrm{lim}^{-1}h_{12}}-1\rv\\\leq&\; \frac{\|h_{12}\|^2}{|h_12^T\Lambda_\mathrm{lim}h_{12}|}\frac{\|\Lambda_\mathrm{lim}^{-1}\|^2\|\Delta\|}{1-\|\Lambda_\mathrm{lim}^{-1}\|\|\Delta\|}\ra 0\,,
\end{align*}
which is the desired limit.
\end{proof}}{Since $$\frac{T(s_0)}{\|T(s_0)\|}=\frac{T(s_0)}{a}\cdot\frac{|a|}{\|T(s_0)\|}\cdot\frac{a}{|a|}\,,$$ 
Due to the space constraints, we refer to~\cite{min2021a} for the complete proof.}
}
\ifthenelse{\boolean{archive}}{
\subsection{Proof of Theorem \ref{thm_unifm_conv_fail}}
\label{app_proof_thm_unifm_conv_fail}
\setcounter{equation}{0}
\def\theequation{\thesubsection.\arabic{equation}}

\begin{proof}
    Let $U(s_0,R):=\{s\in\compl: |s_0-s|<R\}$ be the open ball centered at $s_0$ with radius $R$.
    
    Notice that $z$ is a zero of $\bar{g}(s)$ and it is in the interior of $S$, then given $M>0$, $\exists R_1>0$, s.t. $U(z,R_1)\subset S$, and $\sup_{s\in U(z,R_1)}|\bar{g}(s)|\leq M$. 
    
    Suppose the above is true for some $M>0$, then
    \begin{align*}
        &\;\sup_{s\in S}\lV T(s)-\frac{1}{n}\bar{g}(s)\one\one^T\rV\\ 
        \geq&\;\sup_{s\in U(z,R_1)}\lV T(s)-\frac{1}{n}\bar{g}(s)\one\one^T\rV\\
        \geq&\; \sup_{s\in U(z,R_1)}\lp\|T(s)\|-\sup_{s'\in U(z,R_1)}\lV \frac{1}{n}\bar{g}(s')\one\one^T\rV\rp\\
        =& \; \sup_{s\in U(z,R_1)}\lp\|T(s)\|-\sup_{s'\in U(z,R_1)}|\bar{g}(s')|\rp\\
        \geq& \; \sup_{s\in U(z,R_1)}\|T(s)\|-M\,.
    \end{align*}
    Apparently we only need to find such $M$ and also show that
    \ben
        \exists s^*\in U(z,R_1),\ \text{s.t.} \ \sup_{s\in U(z,R_1)}\|T(s)\|\geq \|T(s^*)\|\geq 2M\,.
    \een
    By the assumption, every $g_i(s)$, $i=1,\cdots,n$ can be written as
    \ben
        g_i(s) = (s-z)h_i(s)\,.
    \een
    Here, $h_i(s)$ is rational and we denote
    \ben
    h_i(z)=h_{i0}\quad i=1,\cdots,n\,,
    \een
    where $h_{i0}\in \mathbb{R}$ and $h_{i0}\neq 0$, $i=1,\cdots,n$. We let $h_\mathrm{max}=\max_{i\in[n]}|h_{i0}|$ and $h_\mathrm{min}=\min_{i\in[n]}|h_{i0}|>0$.
    
    Since for $i\in[n]$, $h_i(s)$ is rational and $z$ is not a pole of $h_i(s)$, $\exists R_2>0$ s.t. $h_i(s)$,$i=1,\cdots,n$ are holomorphic on $U(z,R_2)$, i.e. every $h_i(s)$ has expansion
    \ben
        h_i(s)=h_{i0}+\sum_{k=1}^\infty\frac{h_i^{(k)}(z)}{k!}(s-z)^k:=h_{i0}+r_i(s)\,,
    \een
    where $h_i^{(k)}(\cdot)$ is the $k$-th derivative of $h_i(\cdot)$. For every $i\in [n]$, expand $r_i(s)$ in Taylor series, then using Cauchy's estimation formula~\cite{freitag1977complex}, we have $\forall s\in U(z,\frac{R_2}{2})$,
    \ben
        |r_i(s)|\leq \sum_{k=1}^\infty \frac{M_{h_i}}{R_2^k}|s-z|^k=\frac{M_{h_i}\frac{|s-z|}{R_2}}{1-\frac{|s-z|}{R_2}}\leq \frac{2M_{h_i}}{R_2}|s-z|\,,
    \een
    where $M_{h_i}=\max_{|s-z|=R_2}|h_i(s)|$.
    
    For $s\in U(z,\frac{R_2}{2})$, $\dg\{g_i(s)\}$ can be expanded as
    \ben
        \dg\{g_i(s)\}=(s-z)\dg\{h_i(s)\}=(s-z)(H_0+R(s))\,,
    \een
    where $H_0=\dg\{h_{i0}\}, R(s)=\dg\{r_i(s)\}$. And we also have
    \ben
        \|R(s)\|\leq\max_{i\in[n]}|r_i(s)|\leq  M_h|s-z|\,,
    \een
    where $M_h=\max_{i\in[n]}\frac{2M_{h_i}}{R_2}$.
    
    Now for $T(s)$, we start from
    \begin{align}
        &\;\|T(s)\|\nonumber\\
        =&\;\|(I+\dg\{g_i(s)\}L)^{-1}\dg\{g_i(s)\}\|\nonumber\\
        =&\;|s-z|\|(I+\dg\{g_i(s)\}L)^{-1}\dg\{h_i(s)\}\|\nonumber\\
        \eqref{eq_sv_prod_lb} 
        \geq&\; |s-z|\|(I+\dg\{g_i(s)\}L)^{-1}\|\sigma_1(\dg\{h_i(s)\})\nonumber\\
        =&\; |s-z|\frac{\sigma_1(\dg\{h_i(s)\})}{\sigma_1(I+\dg\{g_i(s)\}L)}\,.
        \label{eq_thm_unifm_conv_fail_1}
    \end{align}
    For $s\in U(z,\frac{R_2}{2})$, the numerator can be lower bounded by
    \begin{align}
        \sigma_1(\dg\{h_i(s)\})&=\;\sigma_1(H_0+R(s))\nonumber\\
      \eqref{eq_weyl_ineq_sv}  &\geq\; h_\mathrm{min}-\|R(s)\|\nonumber\\
      &\geq \; h_\mathrm{min}-M_h|s-z|\,.\label{eq_thm_unifm_conv_fail_2}
    \end{align}
    Then, let 
    \ben L_H=H_0^{1/2}LH_0^{1/2}\,,\een
    which is semi-positive definite. the denominator can be upper bounded by
    \begin{align}
        &\;\sigma_1(I+\dg\{g_i(s)\}L)\nonumber\\
        =&\; \sigma_1(I+(s-z)H_0L+(s-z)R(s)L)\nonumber\\
        =&\;\sigma_1(H_0^{1/2}[ I+(s-z)L_H\nonumber\\
        &\; \quad\quad +(s-z)H_0^{-1/2}R(s)H_0^{-1/2}L_H] H_0^{-1/2})\nonumber\\
        &\;\eqref{eq_sv_prod_ub}\nonumber\\
        \leq&\; \frac{h_\mathrm{max}^{1/2}}{h_\mathrm{min}^{1/2}} \sigma_1\lp I+(s-z)L_H+(s-z)H_0^{-1/2}R(s)H_0^{-1/2}L_H\rp\nonumber\\
        &\; \eqref{eq_weyl_ineq_sv}\nonumber\\
        \leq&\; \frac{h_\mathrm{max}^{1/2}}{h_\mathrm{min}^{1/2}}\lp\sigma_1\lp I+(s-z)L_H\rp+|s-z|\|L_H\|\frac{\|R(s)\|}{h_\mathrm{min}}\rp\,.\label{eq_thm_unifm_conv_fail_3}
    \end{align}
    Suppose $\lambda_n(L_H)=\|L_H\|\geq \max\{\frac{1}{R_1},\frac{2}{R_2}\}$, then for $s^*=z-\frac{1}{\lambda_n(L_H)}\in U(z,\min\{R_1,\frac{R_2}{2}\})$,  we have
    \begin{enumerate}
        \item $I+(s^*-z)L_H=I-\frac{L_H}{\lambda_n(L_H)}$ is symmetric and has eigenvalue 0. Then $\sigma_1(I+(s^*-z)L_H)=0$;
        \item $|s^*-z|\|L_H\|=1$.
    \end{enumerate}
    From \eqref{eq_thm_unifm_conv_fail_3} and \eqref{eq_weyl_ineq_sv}, we have
    \begin{align}
        &\;\sigma_1(I+\dg\{g_i(s^*)\}L)\nonumber\\
        \leq &\;\frac{h_\mathrm{max}^{1/2}}{h_\mathrm{min}^{1/2}}\lp\sigma_1\lp I+(s^*-z)L_H\rp+|s^*-z|\|L_H\|\frac{\|R(s^*)\|}{h_\mathrm{min}}\rp\nonumber\\
        \leq &\; |s^*-z|\|L_H\|\frac{h^{1/2}_\mathrm{max}}{h^{3/2}_\mathrm{min}} \|R(s^*)\|\nonumber\\
        \leq &\;\frac{h^{1/2}_\mathrm{max}}{h^{3/2}_\mathrm{min}}\|R(s^*)\|\leq\frac{h^{1/2}_\mathrm{max}}{h^{3/2}_\mathrm{min}}M_h|s^*-z|\,.\label{eq_thm_unifm_conv_fail_4}
    \end{align}
        
    Combing \eqref{eq_thm_unifm_conv_fail_1}\eqref{eq_thm_unifm_conv_fail_2}\eqref{eq_thm_unifm_conv_fail_4}, let $\lambda_n(L_H)$ large enough such that $\lambda_n(L_H)\geq\max\{\frac{2M_h}{h_\mathrm{min}},\frac{1}{R_1},\frac{2}{R_2}\}$ , we have
    \begin{align*}
        \|T(s^*)\|&\geq\; |s^*-z|\frac{\sigma_1(\dg\{h_i(s^*)\})}{\sigma_1(I+\dg\{g_i(s^*)\}L)}\\
        &\geq \;\frac{h_\mathrm{min}-\frac{M_h}{\lambda_n(L_H)}}{\frac{h^{1/2}_\mathrm{max}}{h^{3/2}_\mathrm{min}}M_h}\geq \frac{h_\mathrm{min}^{5/2}}{2h^{1/2}_\mathrm{max}M_h}\,.
    \end{align*}
    Since $$\lambda_n(L_H)\geq \lambda_n(L)h_\mathrm{min}\geq \lambda_2(L)h_\mathrm{min}\,,$$ by \eqref{eq_sv_prod_lb}, we can have $\lambda_n(L_H)$ arbitrarily large by increasing $\lambda_2(L)$. 
    
    In summary, we can pick $R_2$ for the expansion of $h_i(s)$ to exist, then pick $M=\frac{h^{5/2}_{min}}{4h^{1/2}_\mathrm{max}M_h}$, which also determines $R_1$, and let $$\lambda=\frac{1}{h_\mathrm{min}}\max\lb\frac{2M_h}{h_\mathrm{min}},\frac{1}{R_1},\frac{2}{R_2}\rb\,.$$ Then we conclude that $\forall L$ s.t. $\lambda_2(L)\geq \lambda$, 
    \ben
        \exists s^*=z-\frac{1}{\lambda_n(L_H)}\in U(z,R_1)\subset S,\ \text{s.t.} \ \|T(s^*)\|\geq 2M\,,
    \een
    which implies
    \ben
        \sup_{s\in S}\lV T(s)-\frac{1}{n}\bar{g}(s)\one\one^T\rV\geq \sup_{s\in U(z,R_1)}\|T(s)\|-M\geq M\,.
    \een
\end{proof}
}{}

\subsection{Proof of Lemma \ref{lem_unifm_conv_zero_neighbor}}
\label{app_proof_lem_unifm_conv_zero_neighbor}
\setcounter{equation}{0}
\def\theequation{\thesubsection.\arabic{equation}}

\begin{proof}
    
    We denote $U(s_0,R):=\{s\in\compl: |s_0-s|<R\}$ the open ball centered at $s_0$ with radius $R$.
    
    By assumptions, we have $F_l=\inf_{s\in U(s_0,\delta_0)}f(s)>0$ and $F_h=\sup_{s\in U(s_0,\delta_0)}f(s)<\infty$. Without loss of generality, we assume $F_l=1$. For $s\in U(s_0,\delta_0)$, we have
    \be
        1\leq |f(s)|\leq F_h\label{eq_f_bd_zero}\,.
    \ee
    
    Consider any set $S$, we have
    \ben
        \sup_{s\in S}\lV T(s)-\frac{1}{n}\bar{g}(s)\one\one^T\rV\leq \sup_{s\in S}\|T(s)\|+\sup_{s\in S}|\bar{g}(s)|\,.
    \een
    For the second term, since $\bar{g}(s)$ is rational and $s_0$ is a zero of $\bar{g}(s)$, by continuity of $\bar{g}(s)$, $\exists \delta_1>0$ such that \be \sup_{s\in U(s_0,\delta_1)}|\bar{g}(s)|\leq \frac{\epsilon}{2}\,.\label{eq_bar_g_bd_zero}\ee
    
    Now we bound the first term. To start with, notice that $\mathcal{N}(s_0)=1$. Without loss of generality, assume $g_1(s_0)=0$ and $g_i(s_0)\neq 0, 2\leq i\leq n$. Then again by continuity of $g_i^{-1}(s),2\leq i\leq n$, $\exists \delta_2,M>0$ such that 
    \be
        \sup_{s\in U(s_0,\delta_2)}\max_{2\leq i\leq n}|g_i^{-1}(s)|\leq M\,.\label{eq_tilde_G_inv_bd_zero}
    \ee
    On the other hand, since $s_0$ is a zero of $g_1(s)$, one can pick $\delta_3>0$ such that
    \be
        \sup_{s\in U(s_0,\delta_3)}|g_1(s)|\leq \frac{\epsilon}{36n+2\epsilon nMF_h^2}\,.\label{eq_g1_bd_zero}
    \ee
    We let $\delta = \min\{\delta_0,\delta_1,\delta_2,\delta_3\}$ and we would like to bound $\sup_{s\in U(s_0,\delta)}\|T(s)\|$ by $\frac{\epsilon}{2}$ under sufficiently large $\lambda_2(L)$. This would imply, together with \eqref{eq_bar_g_bd_zero}, that
    \ben
        \sup_{s\in U(s_0,\delta)}\lV T(s)-\frac{1}{n}\bar{g}(s)\one\one^T\rV\leq \epsilon\,.
    \een
    The remaining of this proof is to bound $\sup_{s\in U(s_0,\delta)}\|T(s)\|$ by $\frac{\epsilon}{2}$.
    
    We write $L$ in block form as $\bmt l_{11}&L^T_{21}\\
    L_{21}&\Tilde{L}\emt$, by separating its first row and column from the rest. Here $\Tilde{L}$ is a grounded Laplacian of $L$, and $\tilde{L}$ is invertible as long as $\lambda_2(L)\neq 0$ according to Lemma \ref{lem_grd_Lap_eig_bd}.
    
    Since $L\one_n=0$, we have $L_{21}+\Tilde{L}\one_{n-1}=0$, which gives
    \begin{subequations}\label{eq_lap_bd_zero}
        \begin{equation}
            \Tilde{L}^{-1}L_{21}=-\one_{n-1}\label{eq_lap_bd_zero_1}\,,
        \end{equation}
        \begin{equation}
            L_{21}^T\Tilde{L}^{-1}L_{21}=-L_{21}^T\one_{n-1}=l_{11}\label{eq_lap_bd_zero_2}\,.
        \end{equation}
    \end{subequations}
    We use these equalities later.
    
    We define $\Tilde{G}(s)=\dg\{g_2(s),\cdots,g_n(s)\}$ and $\Tilde{T}(s)=(I_{n-1}+\Tilde{G}(s)f(s)\Tilde{L})^{-1}\Tilde{G}(s)$. Through some computation, we can write $T(s)$ in the following block form
    \ben
        \bmt
        A(s)&-A(s)f(s)L^T_{21}\Tilde{T}(s)\\
        -A(s)f(s)\Tilde{T}(s)L_{21}&\Tilde{T}(s)+A(s)f^2(s)\Tilde{T}(s)L_{21}L^T_{21}\Tilde{T}(s)
        \emt\,,
    \een
    where $$A(s)=\frac{g_1(s)}{1+g_1(s)\lp f(s)l_{11}-f^2(s)L^T_{21}\Tilde{T}(s)L_{21}\rp}\,.$$.
    
    Then an upper bound of $\|T(s)\|$ is given by
    \begin{align}
        &\;\|T(s)\|\nonumber\\
        \leq &\; \lV A(s)\bmt 1&f(s)L^T_{21}\Tilde{T}(s)\\
        f(s)\Tilde{T}(s)L_{21}&f^2(s)\Tilde{T}(s)L_{21}L^T_{21}\Tilde{T}(s)\emt\rV\nonumber\\
        &\;\qquad +\lV \bmt 0&0\\
        0& \Tilde{T}(s)\emt\rV\nonumber\\
        =&\; |A(s)|\lV \bmt1\\f(s)\Tilde{T}(s)L_{21}\emt\bmt1\\f(s)\Tilde{T}(s)L_{21}\emt^T\rV+\|\Tilde{T}(s)\|\nonumber\\
        \leq &\; |A(s)|\lp 1+\|f(s)\Tilde{T}(s)L_{21}\|\rp^2+\|\Tilde{T}(s)\|\,.\label{eq_T_bd_split_zero}
    \end{align}
    Now we bound every term in \eqref{eq_T_bd_split_zero} for $s\in U(s_0,\delta)$ through following steps:
    
    \noindent
    1)\ We firstly bound $|A(s)|$. By Woodbury matrix identity~\cite[0.7.4]{Horn:2012:MA:2422911}, we have
    \begin{align}
        f(s)\Tilde{T}(s)&=\;f(s)\lp I_{n-1}+\Tilde{G}(s)f(s)\Tilde{L}\rp^{-1}\Tilde{G}(s)\nonumber\\
        &=\;\lp\Tilde{L}+f^{-1}(s)\Tilde{G}^{-1}(s)\rp^{-1}\nonumber\\
        &=\; \Tilde{L}^{-1}-\Tilde{L}^{-1}\lp f^{-1}(s)\Tilde{G}(s)+\Tilde{L}^{-1}\rp^{-1}\Tilde{L}^{-1}\,.\label{eq_kv_bd_zero}
    \end{align}
    By \eqref{eq_kv_bd_zero} and \eqref{eq_lap_bd_zero}, we have
    \begin{align}
        &\; l_{11}-f(s)L_{21}^T\Tilde{T}(s)L_{21}\nonumber\\
        =&\; l_{11}-L_{21}^T\lp\Tilde{L}^{-1}-\Tilde{L}^{-1}\lp f^{-1}(s)\Tilde{G}(s)+\Tilde{L}^{-1}\rp^{-1}\Tilde{L}^{-1}\rp L_{21}\nonumber\\
        =&\; (l_{11}-L_{21}^T\Tilde{L}^{-1}L_{21})\nonumber\\
        &\;\quad +L^T_{21}\Tilde{L}^{-1}\lp f^{-1}(s)\Tilde{G}(s)+\Tilde{L}^{-1}\rp^{-1}\Tilde{L}^{-1}L_{21}\nonumber\\
        &\;\eqref{eq_lap_bd_zero_2}\nonumber\\
        =&\; L^T_{21}\Tilde{L}^{-1}\lp f^{-1}(s)\Tilde{G}(s)+\Tilde{L}^{-1}\rp^{-1}\Tilde{L}^{-1}L_{21}\nonumber\\
        &\;\eqref{eq_lap_bd_zero_1}\nonumber\\
        =&\; \one_{n-1}^T\lp f^{-1}(s)\Tilde{G}(s)+\Tilde{L}^{-1}\rp^{-1}\one_{n-1}\,.\label{eq_tf_red_bd_zero}
    \end{align}
    When $\lambda_1(\Tilde{L})\geq 2MF_h$, the following holds:
    \begin{align}
        &\;\lv\one_{n-1}^T\lp f^{-1}(s)\Tilde{G}(s)+\Tilde{L}^{-1}\rp^{-1}\one_{n-1}\rv\nonumber\\
        \leq&\;(n-1)\lV\lp f^{-1}(s)\Tilde{G}(s)+\Tilde{L}^{-1}\rp^{-1}\rV\nonumber\\
        \text{(Lemma \ref{lem_bd_norm_mat_inv})}\leq&\; \frac{n-1}{\sigma_1(\Tilde{G}(s))|f(s)|^{-1}-\|\Tilde{L}^{-1}\|}\nonumber\\
        \eqref{eq_f_bd_zero}\eqref{eq_tilde_G_inv_bd_zero}\leq &\; \frac{n-1}{1/\lp MF_h\rp-1/\lambda_1(\Tilde{L})}\leq 2nMF_h\label{eq_tf_red_bd_zero_part}\,,
    \end{align}
    which when combined with \eqref{eq_g1_bd_zero} gives the following bound on $|A(s)|$:
    \begin{align}
        &\;|A(s)|\nonumber\\
        =&\;\frac{|g_1(s)|}{\lv1+g_1(s)f(s)\lp l_{11}-L^T_{21}f(s)\Tilde{T}(s)L_{21}\rp\rv}\nonumber\\
        &\;\eqref{eq_tf_red_bd_zero}\nonumber\\
        =&\; \frac{|g_1(s)|}{\lv 1+g_1(s)f(s)\one_{n-1}^T\lp f^{-1}(s)\Tilde{G}(s)+\Tilde{L}^{-1}\rp ^{-1}\one_{n-1}\rv}\nonumber\\
        \leq& \; \frac{|g_1(s)|}{1-|g_1(s)||f(s)|\lv\one_{n-1}^T\lp f^{-1}(s)\Tilde{G}(s)+\Tilde{L}^{-1}\rp^{-1}\one_{n-1}\rv}\nonumber\\
        &\;\eqref{eq_g1_bd_zero}\eqref{eq_tf_red_bd_zero_part}\eqref{eq_f_bd_zero}\nonumber\\
        \leq& \; \frac{\frac{\epsilon}{36n+2n\epsilon MF_h^2}}{1-\frac{2\epsilon nMF_h^2}{36n+2n\epsilon MF_h^2}}=\frac{\epsilon}{36 n}\,.\label{eq_T1_bd_zero}
    \end{align}
    
    \noindent
    2) For $\|f(s)\Tilde{T}(s)L_{21}\|$, when $\lambda_1(\Tilde{L})\geq 2M$, we have:
    \begin{align}
        \lV \Tilde{T}(s)L_{21}\rV&=\; \lV\lp\Tilde{G}^{-1}(s)+f(s)\Tilde{L}\rp^{-1}L_{21}\rV\nonumber\\
        &=\; \lV\lp f(s)I+\Tilde{L}^{-1}\Tilde{G}^{-1}(s)\rp^{-1}\Tilde{L}^{-1}L_{21}\rV\nonumber\\
        \eqref{eq_lap_bd_zero_1}&\leq\; \sqrt{n-1}\lV\lp f(s)I+\Tilde{L}^{-1}\Tilde{G}^{-1}(s)\rp^{-1}\rV\nonumber\\
        \text{(Lemma \ref{lem_bd_norm_mat_inv})}&\leq\; \sqrt{n}\lp |f(s)|-\|\Tilde{L}^{-1}\|\|\Tilde{G}^{-1}(s)\|\rp^{-1}\nonumber\\
    \eqref{eq_f_bd_zero} \eqref{eq_tilde_G_inv_bd_zero}   &\leq\; \sqrt{n}\lp 1-M/\lambda_1(\Tilde{L})\rp^{-1}\leq 2\sqrt{n}\leq 3\sqrt{n}-1\,. \label{eq_T2_bd_zero}
    \end{align}
    
    \noindent
    3) Lastly, for $\|\Tilde{T}(s)\|$, when $\lambda_1(\Tilde{L})\geq  M+\frac{4}{\epsilon}$, we have
    \begin{align}
        \|\Tilde{T}(s)\|&=\; \lV\lp\Tilde{G}^{-1}(s)+f(s)\Tilde{L}\rp^{-1}\rV\nonumber\\
        \text{(Lemma \ref{lem_bd_norm_mat_inv})}&\leq\; \lp|f(s)|\lambda_1(\Tilde{L})-\|\Tilde{G}^{-1}(s)\|\rp^{-1}\nonumber\\
        \eqref{eq_f_bd_zero}\eqref{eq_tilde_G_inv_bd_zero}&\leq\;\lp \lambda_1(\Tilde{L})-M\rp^{-1}\leq \frac{\epsilon}{4}\,.\label{eq_T3_bd_zero}
    \end{align}
    
    Apply three bounds obtained from \eqref{eq_T1_bd_zero}\eqref{eq_T2_bd_zero}\eqref{eq_T3_bd_zero} to \eqref{eq_T_bd_split_zero}, we have:
    \ben
        \sup_{s\in U(s_0,\delta)}\|T(s)\|\leq \frac{\epsilon}{36n}\cdot 9n+\frac{\epsilon}{4}=\frac{\epsilon}{2}\,,
    \een
    which holds when $\lambda_1(\Tilde{L})\geq \max\{2MF_h,M+\frac{4}{\epsilon}\}$. According to Lemma \ref{lem_grd_Lap_eig_bd}, We can guarantee it by letting $\lambda_2(L)\geq \lambda=n\max\{2MF_h,M+\frac{4}{\epsilon}\}$. We therefore found the desired $\delta$ and $\lambda$, which finish the proof.
\end{proof}

\ifthenelse{\boolean{archive}}{\subsection{Proof of Lemma \ref{lem_unifm_prob_conv_compact}}
\label{app_proof_lem_unifm_prob_conv_compact}
\setcounter{equation}{0}
\def\theequation{\thesubsection.\arabic{equation}}

\begin{proof}
    It suffices to show that $\forall\epsilon>0$,
    \be\label{eq_lem_unifm_prob_conv_1}
        \lim_{n\ra +\infty}\prob\lp\sup_{s\in S}| \bar{g}_n(s,\mathbf{w})-\hat{g}(s)|\geq\epsilon\rp=0\,,
    \ee
    since $| \bar{g}_n(s,\mathbf{w})-\hat{g}(s)|=\lV \frac{1}{n}\bar{g}_n(s,\mathbf{w})\one\one^T-\frac{1}{n}\hat{g}(s)\one\one^T\rV$.
    
    By the assumptions,  $\{\bar{g}_n(s,\mathbf{w}),n\in\mathbb{N}_+,\mathbf{w}\in\Omega^\infty\}$, and  $\{g_i^{-1}(s,w),i\in\mathbb{N}_+,w\in\Omega\}$ are uniformly bounded by $M_1>0$ and $M_2>0$, respectively on $S$. Then, at any $s\in S$, both $Re\lp g_i^{-1}(s,w)\rp$ and $Im\lp g_i^{-1}(s,w)\rp$ are random variables bounded within $[-M_2,M_2]$. We can simply bound their variances by
    \begin{align*}
        \mathrm{Var}\lp Re\lp g_i^{-1}(s,w)\rp\rp\leq (2M_2)^2=4M_2^2\,,\\
        \mathrm{Var}\lp Im\lp g_i^{-1}(s,w)\rp\rp\leq (2M_2)^2=4M_2^2\,.
    \end{align*}
    Then it follows that
    \begin{align*}
        &\;\mathrm{Var}\lp Re\lp\bar{g}_n^{-1}(s,\mathbf{w})\rp\rp\\ =&\;\mathrm{Var}\lp Re\lp n^{-1}\sum_{i=1}^ng_i^{-1}(s,w)\rp\rp\leq 4M_2^2/n\,,\\\intertext{and}
        &\;\mathrm{Var}\lp Im\lp\bar{g}_n^{-1}(s,\mathbf{w})\rp\rp\\ =&\;\mathrm{Var}\lp Im\lp n^{-1}\sum_{i=1}^ng_i^{-1}(s,w)\rp\rp\leq 4M_2^2/n\,.
    \end{align*}
    By definition of $\hat{g}(s)$ in \eqref{eq_g_bar_stoch}, we have $\expc Re\lp\bar{g}_n^{-1}(s,\mathbf{w})\rp = Re\lp\hat{g}(s)\rp$ and $\expc Im\lp\bar{g}_n^{-1}(s,\mathbf{w})\rp = Im\lp\hat{g}(s)\rp$, then by Chebyshev's inequality, for $\epsilon>0$, we have
    \begin{align}
        &\;\prob\lp \lv\bar{g}_n^{-1}(s,\mathbf{w})-\hat{g}^{-1}(s)\rv\geq \epsilon\rp\nonumber\\
        \leq&\;\prob\lp \lv Re\lp\bar{g}_n^{-1}(s,\mathbf{w})\rp-Re\lp\hat{g}^{-1}(s)\rp\rv+\right.\nonumber\\
        &\;\quad\quad\quad \left. \lv Im\lp\bar{g}_n^{-1}(s,\mathbf{w})\rp-Im\lp\hat{g}^{-1}(s)\rp\rv\geq \epsilon\rp\nonumber\\
        \leq&\; \prob\lp \lv Re\lp\bar{g}_n^{-1}(s,\mathbf{w})\rp-Re\lp\hat{g}^{-1}(s)\rp\rv\geq \epsilon/2\rp+\nonumber\\
        &\; \quad\quad\quad \prob\lp \lv Im\lp\bar{g}_n^{-1}(s,\mathbf{w})\rp-Im\lp\hat{g}^{-1}(s)\rp\rv\geq \epsilon/2\rp\nonumber\\
        \leq&\; \frac{4\mathrm{Var}\lp Re\lp\bar{g}_n^{-1}(s,\mathbf{w})\rp\rp}{\epsilon^2}+\frac{4\mathrm{Var}\lp Im\lp\bar{g}_n^{-1}(s,\mathbf{w})\rp\rp}{\epsilon^2}\nonumber\\
        \leq &\; \frac{32M_2^2}{\epsilon^2n}\,.\label{eq_lem_unifm_prob_conv_2}
    \end{align}
    On the other hand, we have
    \begin{align}
        &\;\lv\bar{g}_n(s,\mathbf{w})\rv\leq M_1\nonumber\\
        \Rightarrow &\;\lv\bar{g}_n^{-1}(s,\mathbf{w})\rv\geq \frac{1}{M_1}\nonumber\\
        \Rightarrow &\;\lv\bar{g}_n^{-1}(s,\mathbf{w})-\hat{g}^{-1}(s)+\hat{g}^{-1}(s)\rv\geq \frac{1}{M_1}\nonumber\\
        \Rightarrow &\;\lv\hat{g}^{-1}(s)\rv\geq \frac{1}{M_1}-\lv\bar{g}_n^{-1}(s,\mathbf{w})-\hat{g}^{-1}(s)\rv\,.\label{eq_lem_unifm_prob_conv_3}
    \end{align}
    Then given $\epsilon>0$, $\forall n\in\mathbb{N}_+, \forall s\in S$, the following holds:
    \begin{align*}
        &\; \prob\lp\lv\hat{g}(s)- \bar{g}_n(s,\mathbf{w})\rv\geq \epsilon\rp\\
        =&\; \prob\lp \lv \bar{g}_n(s,\mathbf{w})\hat{g}(s)\lp\bar{g}_n^{-1}(s,\mathbf{w})-\hat{g}^{-1}(s)\rp\rv\geq \epsilon\rp\\
        \leq &\; \prob\lp \lv \bar{g}_n(s,\mathbf{w})\rv\lv\hat{g}(s)\rv\lv\bar{g}_n^{-1}(s,\mathbf{w})-\hat{g}^{-1}(s)\rv\geq \epsilon\rp\\
        \leq &\; \prob\lp M_1\lv\bar{g}_n^{-1}(s,\mathbf{w})-\hat{g}^{-1}(s)\rv\geq \epsilon|\hat{g}^{-1}(s)|\rp\\
        \eqref{eq_lem_unifm_prob_conv_3}\ \leq &\; \prob\lp M_1\lv\bar{g}_n^{-1}(s,\mathbf{w})-\hat{g}^{-1}(s)\rv\geq\right.\\
         &\; \quad\quad\quad \left. \frac{\epsilon}{M_1}-\epsilon \lv\bar{g}_n^{-1}(s,\mathbf{w})-\hat{g}^{-1}(s)\rv\rp\\
        =&\; \prob\lp \lv\bar{g}_n^{-1}(s,\mathbf{w})-\hat{g}^{-1}(s)\rv\geq \frac{\epsilon}{M_1(M_1+\epsilon)}\rp\\
        \eqref{eq_lem_unifm_prob_conv_2}\ \leq &\; \frac{32M_2^2M_1^2(M_1+\epsilon)^2}{\epsilon^2n}\,.
    \end{align*}
    By taking $n\ra +\infty$ on both sides, we prove that $\bar{g}_n(s,\mathbf{w})$ converges point-wise to $\hat{g}(s)$ on $S$.
    
    We now show that $\bar{g}_n(s,\mathbf{w})$ is also stochastic equicontinuous on $S$. For the definition of stochastic equicontinuity, please refer to~\cite{Newey1991}. We already assumed that $\bar{g}_n(s,\mathbf{w})\leq M_1$, $\forall \mathbf{w}\in \Omega^\infty, s\in S$. Then $\forall \mathbf{w}\in\Omega^\infty, \forall s_1,s_2\in S$, we have
    \begin{align*}
        &\;|\bar{g}_n(s_1,\mathbf{w})-\bar{g}_n(s_2,\mathbf{w})|\\
        \leq &\;\lv\bar{g}_n(s_1,\mathbf{w})||\bar{g}_n(s_2,\mathbf{w})||\bar{g}^{-1}_n(s_1,\mathbf{w})-\bar{g}^{-1}_n(s_2,\mathbf{w})\rv\\
        \leq &\; M_1^2\lv \sum_{i=1}^n\lp g_i^{-1}(s_1,w_i)-g_i^{-1}(s_2,w_1)\rp\rv\\
        \leq &\; M_1^2\sum_{i=1}^n\lv g_i^{-1}(s_1,w_i)-g_i^{-1}(s_2,w_i)\rv\leq nM_1^2L|s_1-s_2|\,,
    \end{align*}
    where the last inequality is from our third assumption and also the fact that $g_i^{-1}(s,w)=g_1^{-1}(s,w)$ (identically distributed as random functions). By~\cite[Corollary 2.2]{Newey1991}, the inequality above is sufficient to establish stochatic equicontinuity of $\bar{g}_n(s,\mathbf{w})$ on $S$, and combining point-wise convergence and the fourth assumption that $\hat{g}(s)$ is uniform continuous, we get the uniform convergence of $\bar{g}_n(s,\mathbf{w})$ to $\hat{g}(s)$ on $S$, which gives \eqref{eq_lem_unifm_prob_conv_1}.
\end{proof}
\subsection{Proofs of Theorem \ref{thm_to_time} and \ref{thm_passive_to_stable}}\label{app_proof_to_time}

\begin{proof}[Proof of Theorem \ref{thm_to_time}]
    When the input to the network is $U(s)$, the output response of the $i$-th node is $$
        Y_i(s)=e_i^TT(s)U(s)\,,
    $$
    where $e_i$ is the $i$-th column of the identity matrix $I_n$.
    
    Using Mellin's inverse formula~\cite[Theorem 3.20]{Dullerud2013}, we have
    \begin{align*}
        &\;|y_i(t)-\bar{y}(t)|\\
        =&\;\lv\frac{1}{2\pi j} \lim_{\omega\ra \infty}\int_{\sigma-j\omega}^{\sigma+j\omega} e^{st}\lp e_i^TT(s)U(s)-e_i^T\bar{g}(s)\one\frac{\one^T}{n}U(s)\rp ds\rv\\
        \leq&\; \frac{e^\sigma}{2\pi}\lim_{\omega\ra \infty}\int_{\sigma-j\omega}^{\sigma+j\omega}\lv e_i^TT(s)U(s)-e_i^T\bar{g}(s)\one\frac{\one^T}{n}U(s)\rv ds\\
        \leq &\;\frac{e^\sigma}{2\pi}\lim_{\omega\ra \infty}\int_{\sigma-j\omega}^{\sigma+j\omega} \lV T(s)-\frac{1}{n}\bar{g}(s)\one\one^T\rV \|U(s)\|ds\\
        =&\;\frac{e^\sigma}{2\pi}\lp (A) + (B) + (C)\rp\,,
    \end{align*}
    where
    $$
        (A)=\int_{\sigma-j\omega_0}^{\sigma+j\omega_0} \lV T(s)-\frac{1}{n}\bar{g}(s)\one\one^T\rV \|U(s)\|ds\,,
    $$
    $$
        (B)=\lim_{\omega\ra \infty}\int_{\sigma+j\omega_0}^{\sigma+j\omega} \lV T(s)-\frac{1}{n}\bar{g}(s)\one\one^T\rV \|U(s)\|ds\,,
    $$
    $$
        (C)=\lim_{\omega\ra \infty}\int_{\sigma-j\omega}^{\sigma-j\omega_0} \lV T(s)-\frac{1}{n}\bar{g}(s)\one\one^T\rV \|U(s)\|ds\,,
    $$
    By our assumption,
    \begin{align*}
        (B)&=\;\lim_{\omega\ra \infty}\int_{\sigma+j\omega_0}^{\sigma+j\omega} \lV T(s)-\frac{1}{n}\bar{g}(s)\one\one^T\rV \|U(s)\|ds\\
        &\leq \; \lim_{\omega\ra \infty}\int_{\sigma+j\omega_0}^{\sigma+j\omega} \lp \lV T(s)\rV+\lV\frac{1}{n}\bar{g}(s)\one\one^T\rV \rp\|U(s)\|ds\\
        &\leq\; 2\gamma\lim_{\omega\ra \infty}\int_{\sigma+j\omega_0}^{\sigma+j\omega} \|U(s)\|ds\leq \frac{2\pi\epsilon}{3e^\sigma \gamma}\,.
    \end{align*}
    Similarly, we have $(C)\leq \frac{2\pi\epsilon}{3e^\sigma \gamma}$.
    
    For the remaining term, we have
    \begin{align*}
        (A)&=\;\int_{\sigma-j\omega_0}^{\sigma+j\omega_0} \lV T(s)-\frac{1}{n}\bar{g}(s)\one\one^T\rV \|U(s)\|ds\\
        &\leq\;\sup_{w\in[-w_0,w_0]}\lV T(\sigma+jw)-\frac{1}{n}\bar{g}(\sigma+jw)\one\one^T\rV\\
        &\;\quad\quad \times\int_{\sigma-j\omega_0}^{\sigma+j\omega_0} \|U(s)\|ds
    \end{align*}
    Since $[\sigma-j\omega_0,\sigma+j\omega_0]$ is a compact set that satisfies the assumption in Theorem \ref{thm_unifm_conv_compact}, we have
    $$
        \lim_{\lambda_2(L)\ra\infty}\sup_{w\in[-w_0,w_0]}\lV T(\sigma+jw)-\frac{1}{n}\bar{g}(\sigma+jw)\one\one^T\rV=0\,.
    $$
    Therefore, for sufficienly large $\lambda_2(L)$, we have $$(A)\leq \frac{2\pi\epsilon}{3e^\sigma \gamma}\,.$$.
    
    Combining the upperbounds for $(A),(B),(C)$, we have
    $$
        |y_i(t)-\bar{y}(t)|\leq \epsilon\,.
    $$
    Notice that the choice of $\lambda_2(L)$ does not depends on time $t$, hence this inequality holds for all $t>0$.
\end{proof}\

\begin{proof}[Proof of Theorem \ref{thm_passive_to_stable}]
    For each $g_i(s), i=1,\cdots,n$, we have, by the OSP property,
    $$
        Re(g_i(s))\geq \epsilon_i |g_i(s)|^2,\forall Re(s)>0\,.
    $$
    Let $\epsilon=\min\{\epsilon_i: i=1\cdots,n\}$, we have
    $$
        Re(G(s))\succeq \epsilon  G^*(s)G(s)\,, 
    $$
    or equivalently,
    $$
        \bmt G(s) \\ I\emt^* \bmt -\epsilon I & I\\ I & 0\emt\bmt G(s) \\ I\emt\succeq 0\,.
    $$
    Since $g_i(s)$ are all OSP, then $g_i(s)$ is positive real~\cite{khalil2002nonlinear}. A positive real function that is not zero function has no zero nor pole on the left half plane. Therefore $g_i(s)$ are invertible for all $Re(s)>0$, which ensures that $G(s)$ is invertible for all $Re(s)>0$.
    Then 
    $$
        (G^*(s))^{-1}\bmt G(s) \\ I\emt^* \bmt -\epsilon I & I\\ I & 0\emt\bmt G(s) \\ I\emt G^{-1}(s)\succeq 0\,,
    $$
    which is
    \be
        \bmt I \\ G^{-1}(s)\emt^* \bmt -\epsilon I & I\\ I & 0\emt\bmt I \\ G^{-1}(s)\emt\succeq 0\,.\label{eq_osp_G}
    \ee
    Notice that $$T(s)=(I+G(s)f(s)L)^{-1}G(s)=(G^{-1}(s)+f(s)L)^{-1}\,,$$
    then from \eqref{eq_osp_G} and the fact that $f(s)$ is PR, we have
    \begin{align*}
        &\;\bmt I \\ T^{-1}(s)\emt^* \bmt -\epsilon I & I\\ I & 0\emt\bmt I \\ T^{-1}(s)\emt\\
        =&\;\bmt I \\ G^{-1}(s)+f(s)L\emt^* \bmt -\epsilon I & I\\ I & 0\emt\bmt I \\ G^{-1}(s)+f(s)L\emt\\
        =&\; \bmt I \\ G^{-1}(s)\emt^* \bmt -\epsilon I & I\\ I & 0\emt\bmt I \\ G^{-1}(s)\emt+[f^*(s)+f(s)]L\\
        \succeq&\; \bmt I \\ G^{-1}(s)\emt^* \bmt -\epsilon I & I\\ I & 0\emt\bmt I \\ G^{-1}(s)\emt\succeq 0\,.
    \end{align*}
    Now for sufficiently large $\gamma>0$, we have
    $$
        \bmt -\epsilon I & I\\ I & 0\emt+\bmt \frac{\epsilon}{2} I & 0\\ 0 & -\gamma^2 \frac{\epsilon}{2} I\emt=\bmt -\frac{\epsilon}{2} I & I\\ I & -\gamma^2 \frac{\epsilon}{2} I\emt\preceq 0\,,
    $$
    since its Schur complement $(-\frac{\epsilon}{2}+\frac{2}{\epsilon \gamma^2})I\preceq 0$ for large $\gamma$.
    
    Therefore,
    \begin{align*}
        &\;\bmt I \\ T^{-1}(s)\emt^* \bmt -\frac{\epsilon}{2} I & 0\\ 0 & \gamma^2 \frac{\epsilon}{2} I\emt\bmt I \\ T^{-1}(s)\emt\\
        \succeq&\;\bmt I \\ T^{-1}(s)\emt^* \bmt -\epsilon I & I\\ I & 0\emt\bmt I \\ T^{-1}(s)\emt\succeq 0\,,
    \end{align*}
    which is exactly,
    $$\gamma^2\frac{\epsilon}{2}(T^{-1}(s))^*(T^{-1}(s))\succeq \frac{\epsilon}{2}I\,.$$
    This shows that
    $$
        \sigma_{min}^2(T^{-1}(s))\geq \frac{1}{\gamma^2},\forall Re(s)>0\,,
    $$
    which is equivalent to
    $$
        \|T(s)\|_2\leq \gamma\,, \forall Re(s)>0\,.
    $$
    
\end{proof}
}{}

\bibliographystyle{IEEEtran}
\bibliography{ref.bib}
\ifthenelse{\boolean{bio}}{
\begin{IEEEbiography}
    [{\includegraphics[width=1in,height=1.25in,clip]{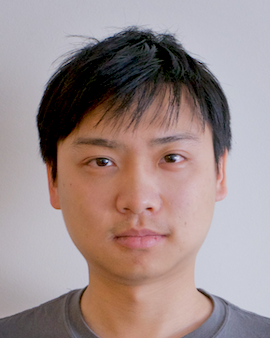}}]{Hancheng Min}
    is currently working toward the Ph.D.
    degree at the Department of Electrical and Computer Engineering, Johns Hopkins University. He received the B.Eng. degree in Electrical Engineering and Automation from Tongji University in 2016, and the M.S. degree in Systems Engineering from University of Pennsylvania in 2018. His research interests include analysis and control for large-scale networks, reinforcement learning and deep learning theory.
\end{IEEEbiography}
\begin{IEEEbiography}
    [{\includegraphics[width=1in,height=1.25in,trim={{2.7cm} 0 {2.7cm} 0},clip]{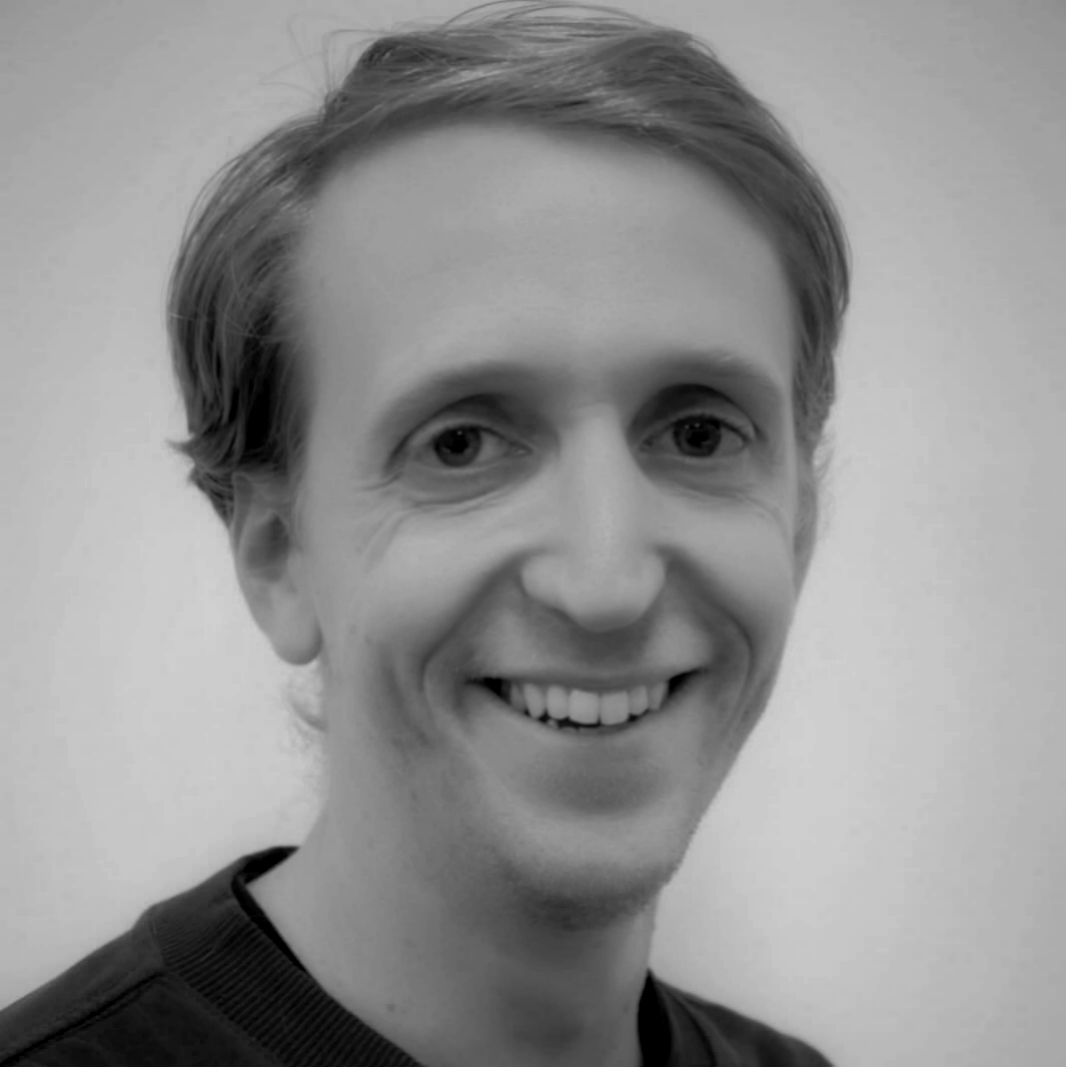}}]{Richard Pates}
    received the M.Eng degree in 2009, and the Ph.D. degree in 2014, both from the University of Cambridge. He is currently an Senior Lecturer at Lund University. His research interests include modular methods for control system design, stability and control of electrical power systems, and fundamental performance limitations in large-scale systems.
\end{IEEEbiography}

\begin{IEEEbiography}[{\includegraphics[width=1in,height=1.25in,clip,keepaspectratio]{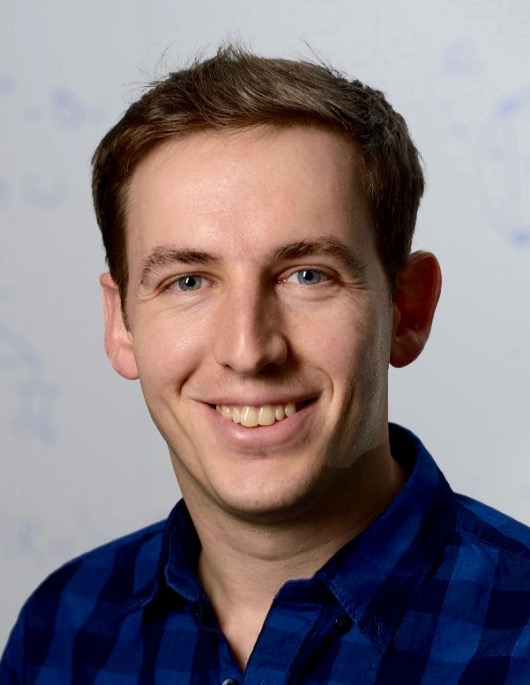}}]
{Enrique Mallada} (S'09-M'13-SM') is an Assistant Professor of Electrical and Computer Engineering at Johns Hopkins University. Prior to joining Hopkins in 2016, he was a Post-Doctoral Fellow in the Center for the Mathematics of Information at Caltech from 2014 to 2016. He received his Ingeniero en Telecomunicaciones degree from Universidad ORT, Uruguay, in 2005 and his Ph.D. degree in Electrical and Computer Engineering with a minor in Applied Mathematics from Cornell University in 2014. 
Dr. Mallada was awarded 
the NSF CAREER award in 2018,
the ECE Director's PhD Thesis Research Award for his dissertation in 2014, 
the Center for the Mathematics of Information (CMI) Fellowship from Caltech in 2014,
and the Cornell University Jacobs Fellowship in 2011. 
His research interests lie in the areas of control, dynamical systems and optimization, with applications to engineering networks such as power systems and the Internet.
\end{IEEEbiography}}{}
\balance

\end{document}